\documentclass{article} 
\usepackage{graphicx, amssymb, color}
\usepackage{amsmath}
\DeclareMathOperator{\sinc}{sinc}
\usepackage{graphicx}
\usepackage{enumerate}
\usepackage{algorithm}
\usepackage{algorithmic}
\usepackage{bbm}

\newtheorem{thm}{Theorem}[section]
\newtheorem{cor}[thm]{Corollary}

\newtheorem{prop}[thm]{Proposition}
\newtheorem{rmk}[thm]{Remark}

\newcommand{\brck}[1]{\left(#1\right)}
\newcommand{\brcksqr}[1]{\left[#1\right]}

\newcommand{\set}[1]{\left\{#1\right\}}
\newcommand{\Real}{\mathbb R}

\newcommand{\prob}{\mathbb{P}}
\newcommand{\E}{\mathbb{E}}
\newcommand{\LT}{\mathcal{L}}
\newcommand{\OperatorA}{\mathcal{A}}
\newcommand{\OperatorG}{\mathcal{G}}
\newcommand{\ILT}{\mathcal{L}^{-1}}
\newcommand{\FPTBM}{\frac{x}{\sqrt{2\pi}}t^{-\frac{3}{2}}e^{-\frac{x^2}{2t}}}

\newenvironment{proof}{\paragraph{Proof:}}{\hfill$\square$}

\usepackage{authblk}

\title{An Economic Bubble Model and Its First Passage Time}
\author{Angelos Dassios\thanks{A.Dassios@lse.ac.uk} }
\author{Luting Li\thanks{L.Li27@lse.ac.uk}}
\affil{Department of Statistics, London School of Economics}

\begin{document}
\date{\today}
\maketitle

\begin{abstract}
	We introduce a new diffusion process $\set{X_t}_{t\ge 0}$ to describe asset prices within an economic bubble cycle. The main feature of the process, which differs from existing models, is the drift term where a mean-reversion is taken based on an exponential decay of the scaled price. Our study shows the scaling factor on $\set{X_t}_{t\ge 0}$ is crucial for modelling economic bubbles as it mitigates the dependence structure between the price and parameters in the model. We prove both the process and its first passage time are well-defined. An efficient calibration scheme, together with the probability density function for the process are given. Moreover, by employing the perturbation technique, we deduce the closed-form density for the downward first passage time, which therefore can be used in estimating the burst time of an economic bubble. The object of this study is to understand the asset price dynamics when a financial bubble is believed to form, and correspondingly provide estimates to the bubble's crash time. Calibration examples on the US dot-com bubble and the 2007 Chinese stock market crash verify the effectiveness of the model itself. The example on BitCoin prediction confirms that we can provide meaningful estimate on the downward probability for asset prices.
	
	\smallskip


\end{abstract}

\section{Introduction}
An \emph{economic bubble} usually refers to economic phenomenons that asset prices extremely deviate from their fundamental values \cite{shiller2000irrational}. One of the most famous bubbles in history, known as the Dutch Tulip Bubble \cite{dash2011tulipomania, garber2001famous}, could be traced back to the 1630s. According to P.M. Garber \cite{garber2001famous}, from November 1636 to February 1637, the prices of tulip bulbs had increased about 20 times. At the peak of the bubble, by selling a few bulbs people could even buy a luxury house in Amsterdam. However, only three months later, the bulbs became worthless. The rapid increases and sudden drops in asset prices are a common feature reflected by a bubble cycle. More modern examples can be found in \cite{wood1992bubble, john2003dot, holt2009summary}.

The \emph{burst} of an economic bubble sometimes follows with financial crisis, or even economic depression. In modern history, the most devastating crisis would be the 2007-2009 Financial Crisis \cite{reinhart20082007}, where people believe the crash of the US real estate market is one of the causing. And the crash itself, is usually referred to as the burst of the US Housing Bubble \cite{holt2009summary}. Although it is believed that a bubble cannot be predicted before it is formed, by knowing the burst time in advance, governments and market participants can manage the potential risk accordingly. Therefore, an effective estimate before the crash will help in preventing systematic risk. The object of this study is to understand the asset price dynamics within an economic bubble cycle and provide estimates to the probability distribution of the collapse time.

The financial bubbles have been studied extensively in econometrics and statistics. As a non-conclusive review, we refer to \cite{phillips2015testing} and the literatures it mentioned for the econometric approach; agent-based models in statistics and a summary of literatures can be found in \cite{filimonov2017modified}. In financial mathematics, local martingale models have been considered in option pricing problems. A. Cox and D.G. Hobson \cite{cox2005local} included a wide branch of stochastic diffusions in their work. 
S. Heston et al. \cite{heston2006options} enriched the discussions by introducing CIR process and Heston stochastic volatility model. In terms of the burst time prediction, C. Brooks and A. Katsaris \cite{brooks2005trading, brooks2005three} forecasted the collapse of speculative bubbles in S\&P 500 index using a three-regime model. To the best of our knowledge, there is limited research in modelling economic bubble dynamics via a pure time-homogeneous diffusion process. The research on finding the explicit probability density of bubble crash time is even less. One paper related to our work is contributed by A. Kiselev and L. Ryzhik \cite{kiselev2010simple}, where a mean-reversion process with an exogenous functional drift has been considered.   

In this paper we introduce a new time-homogeneous diffusion model. Our motivation is to provide an alternative approach, where with the mathematical form to be as simple and tractable as possible, to probabilistically describe asset price dynamics within an economic bubble cycle. The new model is closely linked to the Shiryaev process \cite{shiryaev1961problem, shiryaev2002quickest} derived by A.N. Shiryaev in the context of sequential analysis. Our model involves three independent parameters, where two of them provide mean-reversion effects as in the Ornstein-Uhlenbeck (OU) process. As a crucial variable to our model, the third parameter controls the speed of exponential decay in the drift term. Consequently, the dependence structure among the return, asset price, equilibrium level and the mean-reversion rate has been mitigated. Without introducing extra functionals, our model provides sufficient degree of freedom for calibrations, while on the other hand, avoids over-fitting. Due to the simple structure of the model, we are able to show the closed-form density function of the bubble crash time.

The main contribution in the present paper is that we have provided a self-contained material in modelling bubble dynamic and predicting its burst time. On the theoretical side, we have proved the new model is a well-defined diffusion process and its first passage time (FPT) exists. To be more specific, the process is a semimartingale with a strong and unique solution. As a recurrent strong Markov process, the model embeds an a.s. finite FPT; and its stationary distribution has been found with a neat functional form. On the practical side, a calibration algorithm based on economic features has been considered. We have given explicit solution to the distribution of the process at fixed time. Moreover, the Laplace transform (LT) of the FPT has been found, and based on the perturbation technique we have solved the closed-form density for the downward FPT. In the end, the effectiveness of the model and its FPT density (FPTD) has been verified by three numerical examples.

The rest of the paper is organised as follows: Section 2 introduces the SDE of our new model and the motivation behind it; Section 3 discusses the theoretical results from the new process itself; the closed-form solution of the FPTD is given by Section 4; in Section 5 we demonstrate the calibration algorithm and illustrate the model application via three examples, among which a prediction on the BitCoin collapse time has been given; Section 6 concludes.

\section{Stochastic Dynamic and Motivation}\label{sec2}
Consider a filtered probability space $\set{\Omega, \mathcal{F},\mathbb{P}}$, where $\mathcal{F}=\set{\mathcal{F}_t}_{t\ge 0}$ is a natural filtration generated by a standard Brownian motion $\set{W_t}_{t\ge 0}$. We introduce the following three-parameter SDE
\begin{equation}
\label{eqnexp}
dX_t=\epsilon \left(e^{-2\alpha X_t}-c\right)dt+dW_t,\ X_0=x\in\Real.
\end{equation}
The parameters $\epsilon,\ \alpha$ are restricted on the positive real line and $0\le c\le 1$. 

The dynamic describes a process with exponentially decayed mean-reversion drift. As the most important parameter in our new model, $\alpha$ controls the speed, curvature, and higher order information in the drift term. Figure \ref{fig1} illustrates the functionals of $e^{-2\alpha X_t}$ to different choices of $\alpha$. We can see, as a function of $X_t$, small $\alpha$ produces mildly linear decays in the drift. This extends the range of the process where positive return is maintained. On the other hand, large $\alpha$ generates evident exponential decays. In this case the drift sign is sensitive to the values of $X_t$, and the range of positive drift is compressed.
  
 \begin{figure}[h]
 \centering
 \begin{minipage}[c]{0.48\textwidth}
 \centering
        \includegraphics[width=1.\textwidth, height = .8\textwidth]{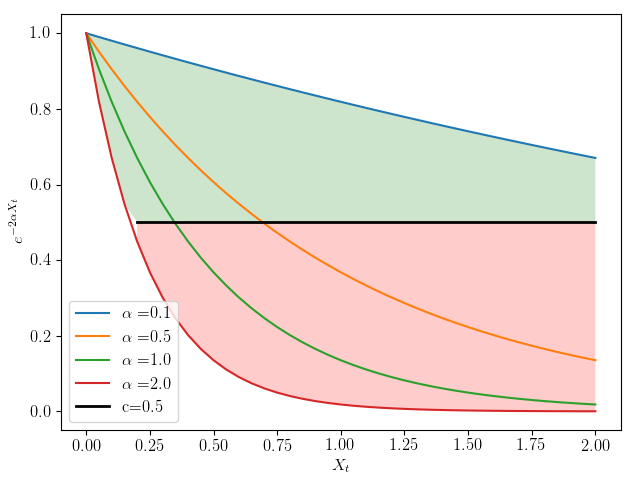}
        \caption{Functoin plots of $e^{-2\alpha X_t}$ with $\alpha=0.1,0.5,1,2$. Green zone: positive drift; red zone: negative drift.}
        \label{fig1}
 \end{minipage}
\hspace{1ex}
\begin{minipage}[c]{0.48\textwidth}
\centering
        \includegraphics[width=1.\textwidth, height = 0.8\textwidth]{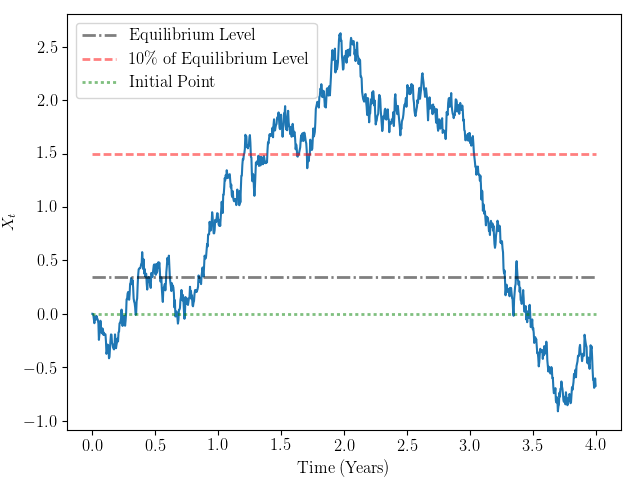}
        \caption{Sample path for $X_t$ in 4 years time. Parameters are chosen as $\alpha=1$, $\epsilon=0.1$, $c=0.5$, $X_0=0$ and $dt=\frac{1}{250}$. }
        \label{fig2}
 \end{minipage}
 \end{figure}

To illustrate the new process in a more intuitive way, we plot the simulated path of $X_t$ with $\alpha=1$ (green curve in Figure \ref{fig1}) in 4 years time. The other parameters are chosen as $\epsilon=0.1$, $c=0.5$ and $X_0=0$. Three thresholds in colors of (from below to above) green, black and red indicate different regimes for the process: I) when $X_t$ is negative or near $0$ (green line), according to SDE \eqref{eqnexp} the process embeds a strong positive trend; II) black line plots the equilibrium level where $e^{-2\alpha X_t}=c$ and in a long-run $X_t$ oscillates around this position; III) red line shows the level of $X_t$ where $e^{-2\alpha X_t}=0.1c$ and the process is forced to drop back due to the strong negative trend. In addition, the sample path shows $X_t$ spends much less time in visiting the equilibrium level from the initial point, than that it drops back from the symmetrical high position. Consider the green curve in Figure \ref{fig1}, this asymmetric feature is a natural reflection to the exponential transforms from the level $X_t$ to the instantaneous return. As a result, $X_t$ in general should have a rapid increase when it is below the equilibrium level, but, even though $X_t$ exceeds the equilibrium level to higher positions, it is not necessary that $X_t$ will drop down immediately. This behaviour essentially differentiates our new model and the OU type mean-reversion processes.

The model feature coincides with observations from economic bubbles. Refer to the theory by H.P. Minsky and H. Kaufman \cite{minsky2008stabilizing}. A bubble cycle is formed by five steps: \emph{Displacement, Boom, Euphoria, Profit Taking, Panic}. In the first step the asset price remains at a lower level and the process usually has an `initiative' to increase. This corresponds to the regime I in our model. During the booming stage, the price becomes sensitive to positive market news and increases rapidly. Although sometimes due to divergence in market anticipations that the price may drop down, after oscillations the asset price will keep increasing. This is described by regime II. In regime III, the peak is shown and large negative drift is accumulated. This describes the euphoria stage where the asset price hits historically high levels; however, due to market capital limits or aversions of risk, the market expectations become negative. In the profit taking stage, asset price becomes sensitive to negative market news and the process shifts from regime III to regime II. In the end the process drops back to the mean-reversion level, or even continue to drop to regime I. This describes the last step of the bubble.

From the calibration point of view, $\alpha$ mitigates the dependence structure of the instantaneous return to the price, equilibrium level and mean-reversion rate. Consider the drift function where $\alpha$ is suppressed, 
\[
\mu(X_t):=\epsilon(e^{-X_t}-c).
\]
In this model once $c$ is determined, the equilibrium level $X_t=-\ln(c)$ becomes a fixed number. If a large rate of $c$ is calibrated, then we simultaneously have a small equilibrium level $-\ln(c)$. Therefore when $X_t$ is small, where $e^{-X_t}-c$ is close to $0$, in order to fit a large instantaneous return\footnote{Otherwise the process will take a long period to visit regime III), where $e^{-X_t}\approx 0$ and $-c$ becomes dominating.} the mean-reversion rate $\epsilon$ should be adjusted highly as well. But we know usually a bubble spends years to finish its whole cycle; so large reversion rate is not desired for a bubble model. 

The analysis shows the $\alpha$-suppressed model is not capable for calibrating a bubble dynamic. As a complement, extra functional term is required (cf. \cite{kiselev2010simple}). However, without introducing extra functions our three-parameter model extends the freedom in model calibration. Combining previous discussions we see SDE \eqref{eqnexp} is a good candidate for describing economic bubbles. 



\section{Theoretical Soundness}
\subsection{Existence, Uniqueness and the Strong Markov Property}
\begin{prop}\label{prop41}
There exists a unique and strong solution $\set{X_t}_{t\ge 0}$ to SDE \eqref{eqnexp} and which has the following explicit form
$$
X_t=x+W_t-c\epsilon t+\frac{1}{2\alpha}\ln\left(1+2\epsilon \alpha e^{-2\alpha x}\int_0^t e^{-2\alpha \left(W_s-c\epsilon s\right)}ds\right);
$$
moreover $\set{X_t}_{t\ge 0}$ is a strong Markov process.
\end{prop}
\begin{proof}
We consider an exponential transform on $X_t$ such that
$$
Y_t = e^{2\alpha X_t}\text{ with } Y_0=e^{2\alpha x}.
$$
By applying Ito lemma we show the coefficients of $Y_t$ satisfy global Lipschitz continuity and linear growth conditions:
\begin{equation}
\label{YtSDE}
dY_t = 2\alpha\left[ \epsilon + (\alpha-c \epsilon) Y_t\right]dt+2\alpha Y_t dW_t.
\end{equation} 
According to Theorem 2.9 in \cite{karatzas2012brownian} we conclude there exists a unique and strong solution $\set{Y_t}_{t\ge 0}$ to SDE \eqref{YtSDE}. Therefore $\set{X_t}_{t\ge 0}$ is the unique and strong solution to SDE \eqref{eqnexp}. On the other hand, refer to \cite[Section 4.4]{talay1994numerical} $Y_t$ has the following explicit form:
$$ 
Y_t = e^{2\alpha \left(W_t-c\epsilon t\right)}\left[Y_0+2\alpha \epsilon \int_0^t e^{-2\alpha \left(W_s-c\epsilon s\right)}ds\right].
$$
Then by substituting $Y_t=e^{2\alpha X_t}$ into the equation above we solve $X_t$. 

Now we consider the strong Markov property. Note that the coefficients in SDE \eqref{YtSDE} are continuous so are bounded on compact subsets of $\mathbb{R}$. Combining the well-posed proof in above, and referring to \cite[Theorem 4.20]{karatzas2012brownian} we show $\set{Y_t}_{t\ge 0}$ is a strong Markov process. Therefore the strong Markov property holds for $\set{X_t}_{t\ge 0}$.
\end{proof}

\begin{rmk}\label{rmk41}
The proof depicts $\set{X_t}_{t\ge 0}$ from another aspect. SDE \eqref{YtSDE} shows $\set{Y_t}_{t\ge 0}$ is a geometric Brownian motion with a mean-reversion drift. Referring to \cite[Equation (9)]{shiryaev1961problem} this is indeed a Shiryaev process. Therefore $\set{X_t}_{t\ge 0}$ is the logarithm of the Shiryaev process. 
\end{rmk}

From the explicit solution in Proposition \ref{prop41} we see $\set{X_t}_{t\ge 0}$ is a semimartingale, where the bounded variation (BV) part consists of a strictly decreasing function and a strictly increasing function. Depending on the Brownian motion path in the exponential integral, for different $t>0$ the BV part could be either positive or negative. However, by observation, when $c=0$ it is clear that only the increasing function is retained. This indicates under special circumstance $\set{X_t}_{t\ge 0}$ could be a submartingale.

\begin{cor}
If $c=0$ then $\set{X_t}_{t\ge 0}$ is a strict submartingale.
\end{cor} 

\begin{proof}
When $c$ is suppressed the solution of $\set{X_t}_{t\ge 0}$ becomes
\begin{equation}\label{eqn:zeroXt}
X_t=x+W_t+\frac{1}{2\alpha}\ln\left(1+2\epsilon \alpha e^{-2\alpha x}\int_0^t e^{-2\alpha W_s}ds\right).
\end{equation}
The adeptness is clear from definition. We consider the $L^1$-integrability of $X_t$. Note that by applying the Jensen's inequality for concave function $\ln(\cdot)$, we have
$$
\mathbb{E}\left[\ln\left(1+2\epsilon \alpha e^{-2\alpha x}\int_0^t e^{-2\alpha W_s}ds\right)\right]\le \ln\left(\mathbb{E}\left[1+2\epsilon \alpha e^{-2\alpha x}\int_0^t e^{-2\alpha W_s}ds\right]\right).
$$
By changing integral and expectation, 
\begin{equation}\label{eqn:deduction1}
\mathbb{E}\left[\int_0^t e^{-2\alpha W_s}ds\right]=\frac{1}{2\alpha^2}\left(e^{2\alpha^2 t}-1\right).
\end{equation}
On the other hand, when $\alpha,\epsilon \in (0,+\infty)$
\begin{equation}\label{eqn:deduction0}
1+2\epsilon \alpha e^{-2\alpha x}\int_0^t e^{-2\alpha W_s}ds\ge 1,\ \forall t\ge 0.
\end{equation}
Therefore 
\begin{equation}\label{eqn:deduction2}
\left|\ln\left(1+2\epsilon \alpha e^{-2\alpha x}\int_0^t e^{-2\alpha W_s}ds\right)\right| = \ln\left(1+2\epsilon \alpha e^{-2\alpha x}\int_0^t e^{-2\alpha W_s}ds\right).
\end{equation}
Combining \eqref{eqn:deduction1} and \eqref{eqn:deduction2} we have
\begin{equation}\label{eqn:deduction3}
\frac{1}{2\alpha}\mathbb{E}\left[\left|\ln\left(1+2\epsilon \alpha e^{-2\alpha x}\int_0^t e^{-2\alpha W_s}ds\right)\right|\right]\le \frac{1}{2\alpha}\ln\left(1+  \frac{\epsilon}{\alpha}e^{-2\alpha x}\left(e^{2\alpha^2 t}-1\right)\right).
\end{equation}
In the end, note that
\begin{equation}\label{eqn:deduction4}
\mathbb{E}\left[\left|W_t\right|\right]=\sqrt{\frac{2t}{\pi}}.
\end{equation}
So applying the triangle inequality and combining \eqref{eqn:deduction3} and \eqref{eqn:deduction4} we show the $L^1$-integrability of $X_t$ by
$$
\mathbb{E}\left[\left|X_t\right|\right]\le x+\sqrt{\frac{2t}{\pi}}+\frac{1}{2\alpha}\ln\left(1+  \frac{\epsilon}{\alpha}e^{-2\alpha x}\left(e^{2\alpha^2 t}-1\right)\right)<+\infty.
$$
Finally, the non-decreasing conditional expectation of $\mathbb{E}\left[X_t\bigg|\mathcal{F}_s\right]$, for $0\le s<t<+\infty$, is given by again using \eqref{eqn:deduction0}, that 
\[
\ln\left(1+2\epsilon \alpha e^{-2\alpha x}\int_0^t e^{-2\alpha W_s}ds\right)>0.
\]
This concludes our proof.
\end{proof}

\subsection{Probability Distribution of $\set{X_t}_{t\ge 0}$}
We consider finding the distribution of $\set{X_t}_{t\ge 0}$. By Proposition \ref{prop41} we see the solution of $X_t$ involves Brownian motion and its exponential integral. Similar problem has been answered by A. Dassios and J. Nagaradjasarma \cite{dassios2006square} for the square-root process. G. Peskir \cite{peskir2006fundamental} deduced the fixed time distribution for the Shiryaev process in a special case. But for the general case only the Laplace transform has been given. Here we refer to the results about Brownian motion and its exponential integral in H. Matsumoto and M. Yor \cite{matsumoto2005exponential, yor1992some}, and have the following proposition.

\begin{prop}\label{prop:density}
For fixed $t>0$ and $u\in \Real$, the probability density of $X_t$ is given by
\[
\mathbb{P}\left(X_t\in du\right)=\alpha du\cdot \left[\int_0^\infty\zeta(u;\frac{c\epsilon}{\alpha},y)\exp\left(-\frac{c^2\epsilon^2 t+1/y+\zeta(u;2,y)}{2}\right) \theta\left(\zeta(u;1,y),\alpha^2 t\right)dy\right],
\] 
where
\[
\theta(r,s)=\frac{r}{\sqrt{2\pi^3s}}e^{\frac{\pi^2}{2s}}\int_0^\infty e^{-\frac{v^2}{2s}-r\cosh(v)}\sinh(v)\sin\left(\frac{\pi v}{s}\right)dv
\]
and
\[
\zeta(u;\mu,y)=\frac{\left(1+2\frac{\epsilon}{\alpha} e^{-2\alpha x}y\right)^{\frac{\mu}{2}}e^{-\mu\cdot \alpha(u-x)}}{y}.
\]
\end{prop}
\begin{proof}
Let $s=\alpha^2 t$, then for another standard Brownian motion $B_s$, with probability $1$ the following equation holds true
\[
B_s+\frac{c\epsilon}{\alpha} s=-\alpha\left(W_t-c\epsilon t\right).
\]
Denote by
\begin{equation}\label{eqn:Bu}
\mu:=\frac{c\epsilon}{\alpha},\ B_s^{(\mu)}:=B_s+\mu s,\ A_s^{(\mu)}:=\int_0^s e^{2B_v^{(\mu)}}dv.
\end{equation}
Then referring to \cite{matsumoto2005exponential, yor1992some} we have
\begin{equation}\label{eqn:XtDensity}
\mathbb{P}\left(A_s^{(\mu)}\in dy,B_s^{(\mu)}\in dz\right)=\frac{1}{y}\exp\left(\mu z-\frac{\mu^2 s}{2}-\frac{1+e^{2z}}{2y}\right)\cdot \theta\left(\frac{e^{z}}{y},s\right)dydz, 
\end{equation}
where
$$
\theta(r,\xi)=\frac{r}{\sqrt{2\pi^3\xi}}e^{\frac{\pi^2}{2\xi}}\int_0^\infty e^{-\frac{v^2}{2\xi}-r\cosh(v)}\sinh(v)\sin\left(\frac{\pi v}{\xi}\right)dv.
$$
On the other hand, re-express $X_t$ using $B_s^{(\mu)}$ and $A_s^{(\mu)}$. Note that
\[
A_s^{(\mu)}=\int_0^{\alpha^2 t}\exp\left(-2\alpha\left(W_{\frac{v}{\alpha^2}}-c\epsilon \frac{v}{\alpha^2}\right)\right)dv.
\]
By changing variable with $w=\frac{v}{\alpha^2}$ we have
\begin{equation}\label{eqn:Au}
A_s^{(\mu)}=\alpha^2\int_0^{ t}e^{-2\alpha\left(W_w-c\epsilon w\right)}dw.
\end{equation}
Rewrite $\set{X_t}_{t\ge 0}$ in Proposition \ref{prop41} using \eqref{eqn:Bu} and \eqref{eqn:Au}, we get
\begin{equation}\label{eqn:Transform}
X_t\overset{\mathbb{P}}{=}x-\frac{B_s^{(\mu)}}{\alpha}+\frac{1}{2\alpha}\ln\left(1+2\frac{\epsilon}{\alpha} e^{-2\alpha x}A_s^{(\mu)} \right).
\end{equation}
We now consider the density function for $X_t$. Note that for fixed $u\in \Real$ and $X_t\le u$, \eqref{eqn:Transform} implies
$$
\alpha(x-u)+\frac{1}{2}\ln\left(1+2\frac{\epsilon}{\alpha} e^{-2\alpha x}A_s^{(\mu)} \right)\le B_s^{(\mu)}.
$$
Denote by
$$
g\left(u,A_s^{(\mu)} \right):=\alpha(x-u)+\frac{1}{2}\ln\left(1+2\frac{\epsilon}{\alpha} e^{-2\alpha x}A_s^{(\mu)} \right).
$$
Considering \eqref{eqn:XtDensity} we have
\begin{equation}
\mathbb{P}\left(X_t\le u\right)= \int_{y\in (0,\infty)}\int_{z\ge g(u,y)}\mathbb{P}\left(A_s^{(\mu)}\in dy,B_s^{(\mu)}\in dz\right).\label{sim:1}
\end{equation}
Taking derivatives on $u$ we further get
\begin{align}
\mathbb{P}\left(X_t\in du\right)&=-\int_{y\in (0,\infty)}\frac{1}{y}\exp\left(\mu g(u,y)-\frac{\mu^2 s}{2}-\frac{1+e^{2g(u,y)}}{2y}\right)\cdot \theta\left(\frac{e^{g(u,y)}}{y},s\right)dy \cdot g^{'}_{u}(u,y)du \\
&=\alpha\int_{y\in (0,\infty)}\frac{1}{y}\exp\left(\mu g(u,y)-\frac{\mu^2 s}{2}-\frac{1+e^{2g(u,y)}}{2y}\right)\cdot \theta\left(\frac{e^{g(u,y)}}{y},s\right)dyd u.\label{eqn:my15}
\end{align}
Introduce the function $\zeta(u;\mu,y)$ that
$$
\zeta(u;\mu,y):=\frac{e^{\mu g(u,y)}}{y}=\frac{\left(1+2\frac{\epsilon}{\alpha} e^{-2\alpha x}y\right)^{\frac{\mu}{2}}e^{-\mu\cdot \alpha(u-x)}}{y}.
$$
Then rearranging \eqref{eqn:my15} we have
\[
\frac{1}{y}\exp\left(\mu g(u,y)-\frac{\mu^2 s}{2}-\frac{1+e^{2g(u,y)}}{2y}\right)\cdot \theta\left(\frac{e^{g(u,y)}}{y},s\right)=\zeta(u;\mu,y)\exp\left(-\frac{\mu^2 s}{2}-\frac{1/y+\zeta(u;2,y)}{2}\right)\cdot \theta\left(\zeta(u;1,y),s\right).
\]
The proof is concluded by substituting $\mu=\frac{c\epsilon}{\alpha}$ and $s=\alpha^2t$ into the equation in above.
\end{proof}

\begin{rmk}
The function $\theta(r,s)$ is closely related to the study in Hartman-Watson distributions \cite{hartman1974normal}. As noticed by H. Matsumoto, M. Yor \cite{matsumoto2005exponential}, and other researchers \cite{barrieu2004study, ishiyama2005methods}, $\theta(r,s)$ is highly oscillating, especially for small $s$. Therefore it is not easy to compute the accurate values of the density. 
\end{rmk}

In practice it is more meaningful to provide the probability distribution function rather than the density function. This requires an extra integral on $\mathbb{P}(X_t\in du)$. Considering the integral involved in $\theta(r,s)$, and the integral taking on $\theta(r,s)$, in total we need to compute three integrals for the distribution function $\mathbb{P}(X_t\le du)$. A direct finite-difference scheme would therefore generate computational efficiency issue. Instead, we consider computing the probability via Monte Carlo simulation.

Referring to Equation \eqref{sim:1} and the density function in Proposition \ref{prop:density}, we have two choices in developing the simulation algorithm. Based on \eqref{sim:1} we could follow the acceptance-rejection approach by considering the relative positions between $z$ and $g(u,y)$. However, in the present paper we will concentrate on the direct sampling scheme by employing the explicit density function.

\begin{prop}\label{def:propFunc}
For fixed $t>0$ and $u\in\Real$, define 
\[
m(z,y)=\alpha\zeta(z;\frac{c\epsilon}{\alpha},y)\exp\left(-\frac{c^2\epsilon^2 t+1/y+\zeta(z;2,y)}{2}\right)\cdot \hat{\theta}\left(\zeta(z;1,y),\alpha^2 t\right),
\]
where
\[
\hat{\theta}(r,s):=\frac{r}{2s}e^{\frac{\pi^2}{2s}}\mathbb{E}\left[Ve^{-r\cosh(V)}\sinh(V)\sinc\left(\frac{V}{s}\right)\right]
\]
with $\sinc(w):=\frac{\sin(w\pi)}{w\pi}$ and $V\sim N(0,\sqrt{s})$. Then for two i.i.d. uniformly distributed random variables $U$ and $Y$, the probability distribution of $X_t$ is given by
\[
\prob\brck{X_t\le u}=\E\brcksqr{\frac{m\brck{-\frac{1}{U}+u+1,\frac{1}{Y}-1}}{U^2Y^2}}.
\]
\end{prop}
\begin{proof}
First we show the identity between $\hat{\theta}(r,s)$ and $\theta(r,s)$. Recall Proposition \ref{prop:density} that
\[
\theta(r,s)=\frac{r}{\sqrt{2\pi^3s}}e^{\frac{\pi^2}{2s}}\int_0^\infty e^{-\frac{v^2}{2s}-r\cosh(v)}\sinh(v)\sin\left(\frac{\pi v}{s}\right)dv.
\]
Rewriting the function we get
\begin{align*}
\theta(r,s)=&\frac{r}{s}e^{\frac{\pi^2}{2s}}\int_0^\infty ve^{-r\cosh(v)}\sinh(v)\sinc\left(\frac{v}{s}\right) \cdot  \frac{1}{\sqrt{2\pi s}}e^{-\frac{v^2}{2s}}dv\\
=&\frac{r}{2s}e^{\frac{\pi^2}{2s}}\mathbb{E}\left[Ve^{-r\cosh(V)}\sinh(V)\sinc\left(\frac{V}{s}\right)\right]\\
=:& \hat{\theta}(r,s).
\end{align*}
Note that the second equation holds true is due to the fact 
\[
(v\sinh(v))\cdot e^{-r\cosh(v)}\cdot \sinc\left(\frac{v}{s}\right)
\]
is an even function.

Now let $m(z,y)$ to be defined as in Proposition \ref{def:propFunc}. Based on the identity between $\theta(r,s)$ and $\hat{\theta}(r,s)$, and referring to Proposition \ref{prop:density}, we can write the probability distribution of $X_t$ as
\begin{equation}
\label{prob:org}
\prob\brck{X_t\le u}=\int_{-\infty}^{u}\int_0^\infty m(z,y)dydz.
\end{equation}
Change variables that
\[
U=-\frac{1}{z-u-1},\ Y=\frac{1}{1+y}.
\]
Then re-expressing $z,y$ by $U,Y$ in \eqref{prob:org} we have
\[
\prob\brck{X_t\le u}=-\int_0^1 \int_1^0m\brck{-\frac{1}{U}+u+1,\frac{1}{Y}-1}\frac{dUdY}{U^2Y^2} .
\]
By noticing the fact that uniform distribution has constant probability density $dU=dY=1$ we conclude the proof.
\end{proof}
\begin{rmk}
The main consideration of involving $\sinc(\cdot)$ is to reduce the oscillation effects from function $\sin(\cdot)$. From the numerical calculation point of view, the function, though cannot totally solve the oscillating issue, could mitigate the chaos to some extent. 
\end{rmk}

\begin{prop} For $t\uparrow +\infty$, the stationary distribution of $X_{\infty}:=\lim_{t\uparrow +\infty}X_t$ is given by
\[
p(x)=\left[\int_{y\in \Real}\frac{w(x)}{w(y)} dy\right]^{-1},
\]
where
\[
w(x)=\exp\set{\frac{\epsilon}{\alpha}e^{-2\alpha x}+2\epsilon cx}.
\]
\end{prop}

\begin{proof}
Consider the Fokker-Planck equation at $t=+\infty$:
\[
\frac{1}{2}p^{''}(x)-\epsilon(e^{-2\alpha x}-c)p^{'}(x)+2\alpha\epsilon e^{-2\alpha x}p(x)=0.
\]
Define $w(x)$ as in the proposition. Solving the ODE without boundary conditions we get
\[
p(x) =\frac{  {C_1}\int_{-\infty}^x w(y) dy+{ C_2} }{w(x) }.
\]
Note that $w(y)\uparrow +\infty$ and is dominated by a double exponential function when $y\downarrow -\infty$. As $w(y)\ge 0$ on $\Real$, so for $x>-\infty$ the integral does not exist:
\[
\int_{-\infty}^x w(y)  dy=+\infty.
\]
In order to get a valid density function we therefore set $C_1=0$. Determining $C_2$ by the full-integrability condition we conclude the proof.

\end{proof}

\begin{rmk}
The stationary distribution is right-skewed due to the fact that the double exponential function diverges much faster than the exponential function. In fixed income modelling, the double exponential function is also used in term structure calibrations (cf.  Nelson-Siegel model \cite{hladikova2012term}). 
\end{rmk}

\subsection{Existence of First Passage Time}
In the later section we will deduce the probability density function of the FPT for $\set{X_t}_{t\ge 0}$. Before conducting the calculations we show the existence of the FPT to any constant level $a\in \Real$.
\begin{prop}\label{prop43}
$\set{X_t}_{t\ge 0}$ is a recurrent process on $\Real$.
\end{prop}
\begin{proof}
Consider the substitution that
\begin{equation}\label{eqn:ZtTransform}
Z_t:=e^{-2\alpha X_t}, \ Z_0=e^{-2\alpha x}.
\end{equation}
Applying the Ito lemma we get
\begin{equation}\label{eqn:ZtSDE}
dZ_t=2\alpha Z_t(-\epsilon Z_t+c\epsilon+\alpha )dt -2\alpha Z_t dW_t.
\end{equation}
Based on Proposition \ref{prop41} we know $\set{Z_t}_{t\ge 0}$ is a diffusion process with unique and strong solution; moreover the strong Markov property holds as well\footnote{In fact one can even find the explicit solution of $Z_t$ by referring to the stochastic Verhulst equation in \cite{talay1994numerical}.}.

Our construction indicates that $\set{Z_t}_{t\ge 0}$ only takes value on the positive half-plane. Additionally by checking with \eqref{eqn:ZtSDE} we see $Z_t=0$ is an absorbing bound. So consider
$$
I=(0,+\infty)
$$
as the domain of $\set{Z_t}_{t\ge 0}$. We show $\set{Z_t}_{t\ge 0}$ is recurrent on $I$ using the scale function as discussed in \cite{karatzas2012brownian}. For any fixed parameter $A\in I$, we define 
\begin{equation}\label{eqn:scaleFunction}
s(z):=\int_A^z \exp\left\{-\int_A^\xi \frac{ (-\epsilon \zeta+c\epsilon +\alpha)}{\alpha \zeta}d\zeta\right\}d\xi,\ z\in I.
\end{equation}
Note that for any $z\in I$ the \emph{nondegeneracy} condition
$$
4\alpha^2 z^2>0
$$
holds. Besides, for the fixed $z\in I$, consider $\delta>0$ is small enough such that $z-\delta \in I$. Then the \emph{local integrability} condition is satisfied by showing
$$
\int_{z-\delta}^{z+\delta}\frac{1+2\alpha \zeta\left|(-\epsilon \zeta +c\epsilon+\alpha)\right|}{4\alpha^2 \zeta^2}d\zeta\le \frac{\delta}{2\alpha^2\left(z^2-\delta^2\right)}+\frac{\epsilon}{\alpha}\delta+\frac{c\epsilon +\alpha}{2\alpha}\ln\left(\frac{z+\delta}{z-\delta}\right)<+\infty.
$$
Therefore the scale function \eqref{eqn:scaleFunction} is well defined. 

We now calculate the limit value of the scale function at boundaries of $I$. Rewrite the $s(z)$ in \eqref{eqn:scaleFunction} as
\begin{equation}\label{eqn:reexpresssz}
s(z)=A^{1+\frac{c\epsilon}{\alpha}}e^{-\frac{A\epsilon }{\alpha}}\int_A^z\frac{e^{\frac{\epsilon }{\alpha}\xi }}{\xi^{1+\frac{c\epsilon}{\alpha}}}d\xi.
\end{equation}
Let $l^{+}:=0^{+}$ and $r^{-}:=\infty^{-}$. Substituting the boundary values into \eqref{eqn:reexpresssz} we get
$$
s(l^{+})=-A^{1+\frac{c\epsilon}{\alpha}}e^{-\frac{A\epsilon }{\alpha}}\int_0^A\frac{e^{\frac{\epsilon }{\alpha}\xi }}{\xi^{1+\frac{c\epsilon}{\alpha}}}d\xi=-\infty,
$$
and 
$$
s(r^{-})=A^{1+\frac{c\epsilon}{\alpha}}e^{-\frac{A\epsilon }{\alpha}}\int_A^{\infty^{-}}\frac{e^{\frac{\epsilon }{\alpha}\xi }}{\xi^{1+\frac{c\epsilon}{\alpha}}}d\xi=+\infty.
$$
Therefore according to Proposition 5.22 in \cite{karatzas2012brownian} we conclude that $\set{Z_t}_{t\ge 0}$ is recurrent on $I$. The recurrence of $\set{X_t}_{t\ge 0}$ on $\Real$ follows by \eqref{eqn:ZtTransform}.
\end{proof}

\begin{cor}\label{cor:FPTExistence}
For any $a, X_0=x\in \Real$, the first hitting time of $\set{X_t}_{t\ge 0}$ from $x$ to $a$ exists. 
\end{cor}

\begin{proof}
This directly follows from Proposition \ref{prop43}.
\end{proof}

\begin{rmk}
Note that in Corollary \ref{cor:FPTExistence} there is no restriction on the direction of FPT. More specifically, define 
\[
\tau_{x\uparrow}^a:=\inf\set{t\ge 0: X_t=a|x<a},
\]
and
\[
\tau_{x\downarrow}^a:=\inf\set{t\ge 0: X_t=a|x>a}
\]
to be the FPTs for from below and from above respectively. Then 
\[
\mathbb{P}\left(\tau_{x\uparrow}^a<+\infty\right|X_0=x)=1
\]
and
\[
\mathbb{P}\left(\tau_{x\downarrow}^a<+\infty|X_0=x\right)=1.
\]
\end{rmk}

\section{Downward First Passage Time of $\set{X_t}_{t\ge 0}$}
By our previous analysis $\set{X_t}_{t\ge 0}$ is a well-defined diffusion process. Therefore the corresponding infinitesimal generator exists. Denote by $C^2$ the collection of twice differentiable and continuous functions defined on $\Real$. For $f\in C^2$, the infinitesimal generator of $\set{X_t}_{t\ge 0}$ is given by
\begin{equation}\label{eqn:generator}
\OperatorA f(x)=\epsilon(e^{-2\alpha x}-c)f^{'}(x)+\frac{1}{2}f^{''}(x).
\end{equation}

\subsection{Dirichlet Problem}\label{sec:41}
Note by our settings the filtration $\mathcal{F}$ is continuous on both sides. So it is equivalent to consider the FPT to either an open or a closed set. \emph{W.l.o.g.}, for $a\in \Real$ we define
\[
\mathcal{D}^{u}:=\set{x\in\Real: x> a},\ \mathcal{D}^{l}:=\set{x\in\Real: x< a}
\]
to be the domains of upper and lower regions to $a$. For notational convenience we denote by $\mathcal{D}$ to refer to either $\mathcal{D}^u$ or $\mathcal{D}^l$. 

Let $\partial \mathcal{D}$ to be the set of boundaries of $\mathcal{D}$. Then 
\[
\partial\mathcal{D}^{u}:=\set{a,+\infty},\ \partial\mathcal{D}^{l}:=\set{a,-\infty}.
\]
The FPT of $\set{X_t}_{t\ge 0}$ with $X_0=x\in \mathcal{D}$ can be defined correspondingly as
\[
\tau_x^a:=\inf\set{t\ge 0: X_t\in \partial \mathcal{D}}.
\]
For a short notation we suppress $x,\ a$ and write 
\[
\tau:=\tau_x^a.
\]

The existence of $\tau$ is given by Corollary \ref{cor:FPTExistence}. In this section we follow G. Peskir and A.N. Shiryaev \cite{peskir2006optimal} to deduce the Dirichlet type boundary value problem for the Laplace transform of $\tau$. Consider an arbitrary sequence of well-defined stopping times $\set{\sigma_n}_{1\le n\le +\infty}$. Define
\[
\sigma:=\lim_{n\uparrow +\infty}\sigma_n.
\]
Note that $\set{X_t}_{t\ge 0}$ is a continuous process. Therefore $\set{X_t}_{t\ge 0}$ is continuous over all stopping times, i.e.
\[
\lim_{n\uparrow +\infty}X_{\sigma_n}=X_\sigma.
\]
Moreover, by Proposition \ref{prop41} $\set{X_t}_{t\ge 0}$ is a strong Markov process. For fixed $\beta \ge 0$, define 
\begin{equation}\label{eqn:XtLT}
f(x)=\mathbb{E}_x\left[e^{-\beta \tau}\right], x\in \mathcal{D},
\end{equation}
where $\mathbb{E}_x[\cdot]:=\mathbb{E}[\cdot | X_0=x]=\mathbb{E}[\cdot | \mathcal{F}_0]$. Then refer to \cite{peskir2006optimal} $f(x,\beta)$ is the unique solution to the following ODE
\begin{equation}\label{eqn:FPTode}
\OperatorA f(x)=\beta f(x),\ x\in\mathcal{D}
\end{equation}
with Dirichlet type boundary conditions
\begin{equation}\label{eqn:boundaries}
f(\partial \mathcal{D})=(1,0)^T.
\end{equation}

\begin{rmk}
The result in above follows from the killed version of potential theory. Note that the strong Markov property and continuity over all stopping times are crucial for representing the unique solution from the Dirichlet problem by \eqref{eqn:XtLT}.
\end{rmk}

\begin{rmk}
The boundary conditions \eqref{eqn:boundaries} imply the boundness in $f$, though the set $C^2$ does not require the boundness explicitly. Therefore by constructing a proper martingale we can use Feynman-Kac theorem to deduce similar conclusion. However, in order to employ the optimal sampling theorem the boundness-related properties in $\tau$ should be further demonstrated. 
\end{rmk}

\subsection{Direct Solution to the Dirichlet Problem}
Let
\begin{equation}\label{eqn:substuts}
\begin{cases}
m:={\frac {
\sqrt { {c}^{2}{\epsilon}^{2}+2 \beta}- \epsilon c}{
2\alpha}},\\
n:={\frac {\sqrt {{c}^{2}{\epsilon}^{2}+2
\beta}+ \alpha}{\alpha}},\\
\psi:={\frac {\epsilon {{ e}^{-2 \alpha x
}}}{\alpha}},\\
\lambda := x(\epsilon c-
\sqrt { {c}^{2}{\epsilon}^{2}+2 \beta}).
\end{cases}
\end{equation}
Refer to \cite{abramowitz1964handbook}. The solution of \eqref{eqn:FPTode}, by substituting \eqref{eqn:generator} into, is given by
\begin{equation}\label{eqn:explicitFPTsolution}
f(x)=C_1e^{\lambda }M\left(m,n,\psi\right)+ C_2 e^{\lambda}U\left(m,n,\psi\right),
\end{equation}
where $M(m,n,\psi)$ and $U(m,n,\psi)$ are solutions to the Kummer's equation \cite{daalhuis2010confluent} 
\[
\psi u^{''}(\psi)+(n-\psi)u^{'}(\psi)=a u(\psi).
\]
Now determine the constants $C_1$ and $C_2$. Consider the hitting from below case, i.e. the boundary is taken on $\partial \mathcal{D}^l$. Substituting $x=-\infty$ we see
\[
\psi=+\infty
\]
and 
\[
\lambda = +\infty
\]
as $\epsilon c-\sqrt { {c}^{2}{\epsilon}^{2}+2 \beta}<0$. Refer to \cite{lozier2003nist}, the asymptotic of $U\left(m,n,\psi\right)$ for large $\psi$ is given by 
\[
U\left(m,n,\psi\right)\sim \psi^{-m},\ \psi\uparrow +\infty.
\]
Although by $m>0$, $U\left(m,n,\psi\right)$ converges to $0$ for large $\psi$, $e^\lambda U\left(m,n,\psi\right)$ still diverges when $\lambda \uparrow +\infty$. On the other hand, referring to \cite{lozier2003nist} again we have
\[
M\left(m,n,\psi\right)\sim \frac{e^{\psi}\psi^{m-n}}{\Gamma(m)},\ \psi\uparrow +\infty.
\]
Therefore the limit value at $x=-\infty$ does not exist for either $e^\lambda M\left(m,n,\psi\right)$ or $e^\lambda U\left(m,n,\psi\right)$. The unique solution for the hitting from below case then becomes
\[
f(x)\equiv 0.
\]
This indicates the LT for the upward FPT does not really exist. 

On the other hand, consider the solution to FPT from above, i.e. a downward hitting time that the boundary is taken on $\partial \mathcal{D}^u$. By substituting $x=+\infty$ we get
\[
\psi=0\ \text{and}\ \lambda = -\infty.
\]
Refer to \cite[Section 13.2 (iii)]{lozier2003nist}, depending on the choices of $\beta$ the limit of $U\left(m,n,\psi\right)$ at $\psi = 0^{+}$ has various versions. In order to guarantee the uniqueness of solution\footnote{The parameter $\beta$ is involved in the LT. We want a function of solution $f(x)$ that is unique in functional forms to all $\beta\ge 0$.} we set $C_2=0$. For the limit of $M\left(m,n,\psi\right)$ we have
\[
M\left(m,n,\psi\right)=1+O(\psi),\ \psi\downarrow 0.
\]
Therefore the $+\infty$ boundary gives the solution 
\[
f(x)=C_1 e^\lambda M\left(m,n,\psi\right).
\]
Consider the boundary condition on $x=a$. We write
\[
\begin{cases}
\hat{\psi}:={\frac {\epsilon {{ e}^{-2 \alpha a
}}}{\alpha}},\\
\hat{\lambda}:= a(\epsilon c-
\sqrt { {c}^{2}{\epsilon}^{2}+2 \beta}).
\end{cases}
\]
Then $f(a)=1$ gives 
\begin{equation}\label{enq:actTransformSolution}
f(x)=\frac{e^{\lambda}M\left(m,n,\psi\right)}{e^{\hat{\lambda}}M\left(m,n,\hat{\psi}\right)}.
\end{equation}

\begin{rmk}
As indicated by our analysis, only a downward LT for the present problem exists. In practice we are more interested in the burst time of an economic bubble rather than predicting how record-high would the bubble visit. Therefore the missing solution in upward LT would be a minor issue. 
\end{rmk}


Equation \eqref{enq:actTransformSolution} shows the LT for the downward first hitting time. Due to the special function it is difficult to find the explicit inverse transform. For numerical inversion schemes we refer to \cite{abate2006unified}, where three efficient algorithms are provided. However, considering the complicated functional form, it can be imagined that the speed and accuracy in the numerical inverse may not be desired. 

\subsection{Perturbed FPTD}\label{sec43}
The earliest and most successful application of the perturbation technique could be traced back to in finding the solutions of the Schrodinger equation for Hamiltonians of even moderate complexity \cite{schrodinger1926quantisierung,soliverez1981general}. In mathematical finance, perturbation theory has been studied extensively; see \cite{fouque2011multiscale, dassios2010perturbed, duck2009singular}. Inspired by A. Dassios and S. Wu \cite{dassios2010perturbed}, J. Fouque et al. \cite{fouque2011multiscale}, we apply perturbations on the mean-reversion parameter $\epsilon$ and find the closed-form density for the downward FPT of $\set{X_t}_{t\ge 0}$.

\emph{W.l.o.g.} we let $a=0$ and the FPT problem is defined on $\mathcal{D}^{u}$. Consider the function $f\in C^2$. Assume there exists a sequence of $C^2$ functions $\set{f_i}_{i\ge 0}$, such that
\begin{equation}\label{eqn:perturbation}
f=\sum_{i=0}^{\infty}\epsilon^i f_i.
\end{equation}
Substitute \eqref{eqn:perturbation} into \eqref{eqn:FPTode}:
\begin{equation}\label{eqn:ptbsubs}
\sum_{i=0}^\infty\epsilon^i \OperatorA f_i=\sum_{i=0}^{\infty }\epsilon^i \beta f_i.
\end{equation}
For any $f\in C^2$, introduce
\[
\OperatorG f(x):=\frac{1}{2}f^{''}(x).
\]
Rearranging the terms in \eqref{eqn:ptbsubs} we have
\[
\OperatorG f_0-\beta f_0+\sum_{i=1}^{\infty}\epsilon^i\left(\OperatorG f_i-\beta f_i +(e^{-2\alpha x}-c)f_{i-1}^{'}\right)=0.
\]
By assigning proper boundary conditions we split the original Dirichlet problem \eqref{eqn:FPTode} and \eqref{eqn:boundaries} into recursive representations:
\begin{equation}\label{eqn:o1}
o(1):\ \OperatorG f_0-\beta f_0=0,\ f_0(\partial D)=(1,0)^T,
\end{equation}
and for $i\ge 1$
\begin{equation}\label{eqn:o2}
o(\epsilon^i):\ \OperatorG f_i-\beta f_i +(e^{-2\alpha x}-c)f_{i-1}^{'}=0,\ f_i(\partial D)=(0,0)^T.
\end{equation}

\begin{rmk}\label{rmk:45}
The $o(1)$ problem in fact is the corresponding boundary value problem for the downward FPT of Brownian motion. Introduce 
\[
\tau_W:=\inf\set{t\ge 0: W_t=0|W_t=x>0}.
\]
Then $f_0(x)=\mathbb{E}_x\left[e^{-\beta \tau_W}\right]$. In addition, for $i\ge 1$ the function $f_i$ embeds the following representation \cite{peskir2006optimal}
\[
f_i(x)=\mathbb{E}_x\left[\int_0^{\tau_W}e^{-\beta s}\left(e^{-2\alpha W_s}-c\right)f_{i-1}^{'}(W_s)ds\right].
\]
According to \cite{peskir2006optimal} the solutions for $o(1)$ and $o(\epsilon^i)$, $i\ge 1$, are unique. Therefore the existence of $\set{f_i}_{i\ge 1}$ is guaranteed and the perturbation representation \eqref{eqn:perturbation} is valid.
\end{rmk}

In the present paper we solve the recursive system up to $i=1$ and provide the $o(\epsilon)$-accurate FPTD estimation. Referring to \cite{borodin2012handbook}, for $o(1)$ we have
\begin{equation}\label{eqn:bmFPTLT}
f_0(x)=e^{-\gamma x},
\end{equation}
where $\gamma:=\sqrt{2\beta}$. Further let $f_1=f_0g_1$. Then solving $o(\epsilon)$ we get\begin{equation}
\label{fptabove}
g_1(x)=\frac {\gamma \left( {{ e}^{-2\alpha x}}-1 \right) }{
2\alpha\left( \gamma+\alpha \right) }+cx.
\end{equation}

\begin{prop}\label{cor42}
Let $\tau^*$ to be the first order approximation of $\tau$. Then the FPTD of $\tau^{*}$ is given by
\[
\mathbb{P}^1_x(\tau^*\in dt)=\left(1+\epsilon \left(cx+ \frac{\left(1-e^{-2\alpha x}\right)(\alpha t-x)}{2\alpha x}\right)\right)p_0(t)-\epsilon \frac{\alpha}{4}\left(1-e^{-2\alpha x}\right){{e}^{\alpha x \left(  \frac{\alpha t}{2x}+1 \right) }}{
\text{Erfc}}\left(\frac{x}{\sqrt{2t}}+\alpha \sqrt {\frac{t}{2}}
 \right),
\] 
where $\mathbb{P}_x(\cdot)=\mathbb{P}(\cdot|X_0=x)=\mathbb{P}(\cdot|\mathcal{F}_0)$ and $p_0(t)$ is the downward FPTD for Brownian motion
$$
p_0(t)=\frac{x}{\sqrt{2\pi}}t^{-\frac{3}{2}}e^{-\frac{x^2}{2t}}.
$$
$\text{Erfc}(\cdot)$ is the complementary error function given by
\[
\text{Erfc}(z)=\frac{2}{\sqrt{\pi}}\int_z^\infty e^{-y^2}dy.
\]
\end{prop}

\begin{proof}
Refer to \eqref{eqn:perturbation}, \eqref{eqn:bmFPTLT} and \eqref{fptabove}. The first order perturbed LT of the downward FPT is given by
\begin{equation}\label{eqn:totaltransform}
f^1(x)=f_0(x)\left(1+\epsilon cx\right)+\epsilon f_0(x) \frac{\gamma \left(e^{-2\alpha x}-1\right)}{2\alpha(\alpha+\gamma)}.
\end{equation}
Note that $\gamma=\sqrt{2\beta}$. Therefore $f^1(x)$ is a function of $\beta$ as well. To emphasize the transform parameter we denote by $f^1(\beta)$ and $f_0(\beta)$ respectively. As mentioned by Remark \ref{rmk:45}, $f_0(\beta)$ is nothing but the LT for the downward FPT of Brownian motion. According to \cite{bateman1954tables} this gives
\[
p_0(t):=\ILT\set{f_0(\beta)}(t)=\FPTBM.
\]

Now consider the inverse transform for the second term in \eqref{eqn:totaltransform}. Define 
\begin{equation}\label{sTrans}
\tilde{l}_1(\beta):=\frac{\sqrt{\beta} e^{-\sqrt{\beta}}}{\alpha x+\sqrt{\beta}}.
\end{equation}
Then the second term can be re-written as 
\begin{equation}
\label{strans2}
 f_0(\beta) \frac{\gamma \left(e^{-2\alpha x}-1\right)}{2\alpha(\alpha+\gamma)}=\frac{e^{-2\alpha x}-1}{2\alpha } \cdot \tilde{l}_1(2x^2 \beta).
\end{equation}
Refer to \cite{bateman1954tables}. The inverse of \eqref{sTrans} is given by
\begin{equation}\label{invtrans2}
\ILT\set{\tilde{l}_1(\beta)}(t)={\alpha}^{2}{x}^{2}{{ e}^{\alpha x \left( \alpha x t+1 \right) }}{
 \text{Erfc}} \left( {\frac {1}{2\sqrt {t}}}+\alpha x\sqrt {t}
 \right) -{\frac {2 \alpha x t-1}{2\sqrt {\pi}}{t}^{-\frac{3}{2}}{{ e}^
{-\frac{1}{4t}}}}.
\end{equation}
According to the property of inverse Laplace transform \cite{oberhettinger2012tables}, for constant $c$ 
\[
\ILT \set{\frac{1}{c}\tilde{l}_1\left(\frac{\beta}{c}\right)}(t)=\ILT\set{\tilde{l}_1(\beta)}(ct).
\]
So let $c=\frac{1}{2x^2}$ we have
\begin{equation}\label{eqn:secondPart}
\ILT\set{\tilde{l}_1(2x^2\beta)}(t)=\frac{1}{2x^2}\ILT\set{\tilde{l}_1(\beta)}\left(\frac{t}{2x^2}\right).
\end{equation}
Summarizing \eqref{strans2}, \eqref{invtrans2}, \eqref{eqn:secondPart} gives the inverse transform for the second term in \eqref{eqn:totaltransform}. This concludes the proof.
\end{proof}

\begin{rmk}\label{rmk:47}
As an approximation, the first order perturbation provides a continuous function but not necessarily a valid probability density function. In fact,
\[
\mathbb{E}_x\left[\tau^*\right]=\lim_{\beta \downarrow 0}f^{1}(\beta)=1+\epsilon  c x.
\]
In the case $c>0$ the first order perturbation would provide an extra tiny probability by $\epsilon cx$. We will discuss the accuracy issue in the later proposition.
\end{rmk}

Follow directly with Proposition \ref{cor42}. The tail asymptotics and probability distribution of running minimum are given explicitly.

\begin{cor}\label{cor:asymptotics}
The tail asymptotics for $\mathbb{P}^1_x\left(\tau^{*}\in dt\right)$ are given by
\begin{equation}
\tag{\text{Left Tail Asymptotics}}
\mathbb{P}^1_x(\tau^{*}\in dt)\sim \left(1+\epsilon \left(cx- \frac{1-e^{-2\alpha x}}{2\alpha}\right)\right)p_0(t)\sim p_0(t),\ t\downarrow 0^{+},
\end{equation}
and
\begin{equation}
\tag{\text{Right Tail Asymptotics}}
\mathbb{P}^1_x(\tau^{*}\in dt)\sim \left(1+\epsilon \left(cx+\frac{(1-\alpha x)\left(1-e^{-2\alpha x}\right)}{2\alpha ^2 x}\right)\right)p_0(t)\sim p_0(t),\ t\uparrow +\infty.
\end{equation}
\end{cor}

\begin{proof}
The left tail asymptotic is given by calculations referring to \cite{chiani2003new}
\begin{equation}\label{eqn:asympTails}
\text{Erfc}\left(\frac{x}{\sqrt{2t}}+\alpha \sqrt {\frac{t}{2}}\right)\sim \exp\set{-\left(\frac{x}{\sqrt{2t}}+\alpha \sqrt{\frac{t}{2}}\right)^2},\ t\downarrow 0^+.
\end{equation}
Consider the right tail. Note that if we repeat using \eqref{eqn:asympTails}, the second term of $\mathbb{P}^1_x(\tau^{*}\in dt)$ will remain as a constant while the first term vanishes. This leads to a constant tail asymptotic for $t\uparrow +\infty$, however, $\mathbb{P}^1_x(\tau^{*}\in dt)\downarrow 0$ indeed. Refer to another fact \cite{olver2010nist} that
\begin{equation}\label{erfcAsym2}
\text{Erfc}(y)\sim \frac{e^{-y^2}}{y\sqrt{\pi}}\left(1-\frac{1}{2y^2}\right),\ y\uparrow +\infty.
\end{equation}
Also note the first term of $\mathbb{P}^1_x(\tau^{*}\in dt)$ can be re-expressed as
\begin{equation}\label{eqn:splitPt}
\left(1+\epsilon \left(cx+ \frac{\left(1-e^{-2\alpha x}\right)(\alpha t-x)}{2\alpha x}\right)\right)p_0(t)=\left(1+\epsilon \left(cx- \frac{\left(1-e^{-2\alpha x}\right)}{2\alpha }\right)\right)p_0(t)+\epsilon  \frac{\left(1-e^{-2\alpha x}\right) t}{2 x}p_0(t).
\end{equation}
Then substituting \eqref{erfcAsym2} and \eqref{eqn:splitPt} into $\mathbb{P}^1_x(\tau\in dt)$, we find as $t\uparrow +\infty$,
\begin{align*}
\mathbb{P}^1_x(\tau^{*}\in dt)\sim&  \left(1+\epsilon \left(cx- \frac{\left(1-e^{-2\alpha x}\right)}{2\alpha }\right)\right)p_0(t)+\epsilon\frac{\left(1-e^{-2\alpha x}\right)e^{-\frac{x^2}{2t}}}{2\sqrt{2\pi t}}\\
&-\epsilon \frac{\alpha\left(1-e^{-2\alpha x}\right)}{4}\cdot\frac{\sqrt{2}e^{-\frac{x^2}{2t}}}{\sqrt{\pi t}\alpha}+\epsilon \frac{\alpha\left(1-e^{-2\alpha x}\right)}{4}\cdot\frac{2\sqrt{2}e^{-\frac{x^2}{2t}}}{2\sqrt{\pi}\alpha^3 t^{\frac{3}{2}}}\\
=&\left(1+\epsilon \left(cx- \frac{\left(1-e^{-2\alpha x}\right)}{2\alpha }\right)\right)p_0(t)+0+\epsilon \frac{1-e^{-2\alpha x}}{2x\alpha^2}p_0(t).
\end{align*}
This completes the proof.
\end{proof}

\begin{rmk}
Corollary \ref{cor:asymptotics} indicates the FPTD of $\tau^{*}$ has the same tails as the FPTD of Brownian motion. We know Brownian motion is a null-recurrent Markov process. Therefore we may infer $\mathbb{E}_{x}\left[\tau^*\right]=+\infty$. Indeed, according to the first moment rule and by \eqref{eqn:totaltransform} we can check 
\[
\mathbb{E}_x\left[\tau^*\right]=-\frac{\partial f^1(\beta)}{\partial \beta}\bigg|_{\beta = 0}=+\infty.
\]
\end{rmk}

\begin{cor}\label{cor:runningMinimum}
For fixed $t\ge 0$ and $a<x$, denote the running minimum by
\[
X_t^*:=\min_{0\le u\le t}\set{X_u}.
\]
Also write $\mathbb{P}_x^1\left(\tau^{*}\in dt\right)$ with parameters $\epsilon,\ x,\ \alpha,\ c$ as
\[
p^1\left(t|\epsilon, x, \alpha, c\right).
\]
Then the first order perturbed distribution of $X_t^*$ is given by
\[
\mathbb{P}_x^1\left(X_t^*\le a\right)=\int_{0}^t p^1\left(u|\epsilon e^{-2\alpha a}, x-a, \alpha, c e^{2\alpha a}\right) du.
\]
\end{cor}

\begin{proof}
Introduce $\set{Y_t}_{t\ge 0}$ such that
\[
Y_t=X_t-a,\ \forall t\ge 0.
\]
The SDE of $Y_t$ is given by
\begin{equation}\label{eqn:YTSDEANOTHER}
dY_t=e^{-2\alpha a}\epsilon\brck{e^{-2\alpha Y_t}-ce^{2\alpha a}}dt+dW_t,\ Y_0=y=x-a.
\end{equation}
Similarly define $Y_t^*:=\min_{0\le u\le t}\set{Y_u}$. Then
\begin{equation}\label{eqn:Euilvalence1}
\mathbb{P}_x\brck{X_t^*\le a}=\mathbb{P}_{y}\brck{Y_t^*\le 0}.
\end{equation}
On the other hand, let $\tau$ to be the FPT of $Y_t$ from $y$ to $0$. Note the fact that
\begin{equation}\label{eqn:equivalienceTRSNAFORM}
\mathbb{P}_y\brck{Y_t^*\le 0}=\mathbb{P}_y\brck{\tau\le t}.
\end{equation}
So substituting the parameters in \eqref{eqn:YTSDEANOTHER} into Proposition \ref{cor42}, and considering the equivalence between \eqref{eqn:equivalienceTRSNAFORM} and \eqref{eqn:Euilvalence1}, we prove the result.
\end{proof}

Now we consider the accuracy of the perturbation estimation. Denote the actual FPTD of $\tau$ by
\[
\mathbb{P}_x\left(\tau\in dt\right).
\]
Let the absolute error to be denoted by
\begin{equation}\label{eqn:errorfunction}
q_\tau(t):=\mathbb{P}_x\left(\tau\in dt\right) - \mathbb{P}^1_x\left(\tau^*\in dt\right).
\end{equation}
Then we show $q_\tau(t)$ is $o(\epsilon)$-accurate.

\begin{prop}\label{prop:411}
For any $t>0$, there exists a constant $M>0$ such that 
\[
\left|q_\tau(t)\right|\le  M\epsilon^2.
\]
Moreover, the probabilistic representation of $q_\tau(t)$ is given by
$$
q_\tau(t)=\epsilon^2\mathbb{E}_x\left[\int_0^{t\wedge \tau}\left(e^{-2\alpha X_u}-c\right)\eta(t-u,X_u)du\right],
$$
where 
\[
\eta(t,x)=-\frac{\alpha^2 \cosh(\alpha x)}{2}M_1(t,x)+\frac{1-e^{-2\alpha x}}{2\sqrt{2\pi}\alpha }M_2(t,x)+\frac{e^{-2\alpha x}}{\sqrt{2\pi}}M_3(t,x)+c\left(2-\frac{x^2}{t}\right)p_0(t)
\]
with
\[
\begin{cases}
M_1(t,x)=\text{Erfc}\left(\frac{x}{\sqrt{2t}}+\alpha  \sqrt{\frac{t}{2}}\right)e^{\frac{\alpha^2}{2}t}\\
M_2(t,x)=e^{-\frac{x^2}{2t}}\left[\alpha^2 t^2 -(\alpha x+1)t+x^2\right]t^{-\frac{5}{2}}\\
M_3(t,x)=e^{-\frac{x^2}{2t}}\left(\alpha t -x\right)t^{-\frac{3}{2}}
\end{cases}.
\]
\end{prop}

\begin{proof}
To emphasize the dual effects of $f_1$ as a function of $x,\beta$, denote by $f_1(\beta, x):=f_1(x)=f_1(\beta)$. Let $\eta(t,x):=\ILT\set{\frac{\partial}{\partial x}f_1(\beta,x)}(t)$ to be the inverse transform of $f_1^{'}(x)$. Consider using same tricks in the proof of Proposition \ref{cor42}. After standard calculations we show $\eta(t,x)$ embeds the explicit form as in above.

Now we prove the probabilistic representation and uniform boundness. Let $h(x):=(e^{-2\alpha x}-c)$. By the solution of $\eta(t,x)$ we know
\[
\lim_{t\uparrow +\infty}\left|h(x)\eta(t,x)\right|=0,\ \forall x\in\mathcal{D}^{u}.
\]
On the other hand, according to Corollary \ref{cor:FPTExistence}, $\mathbb{P}_x\left(\tau<+\infty\right)=1$. Therefore
\[
\lim_{t\uparrow +\infty}\int_0^{t\wedge \tau}\left|h(X_u)\eta(t-u,X_u)\right|du=\int_0^{ \tau}\lim_{t\uparrow +\infty}\left|h(X_u)\eta(t-u,X_u)\right|du=0.
\]
Since $\int_0^{t\wedge \tau}\left|h(X_u)\eta(t-u,X_u)\right|du$ is continuous on $t$ and 
\[
\lim_{t\downarrow 0}\int_0^{t\wedge \tau}\left|h(X_u)\eta(t-u,X_u)\right|du = 0,
\]
so there exists $M>0$ such that
\begin{equation}\label{eqn:boundnessCondition}
\int_0^{t\wedge \tau}\left|h(X_u)\eta(t-u,X_u)\right|du \le M,\ \forall t\ge 0,\text{ and }\set{X_u}_{0\le u\le t}\in \mathcal{D}^u.
\end{equation}
Let $\tilde{q}_\tau(t)$ defined by
\[
\tilde{q}_\tau(t)=\epsilon^{2}\mathbb{E}_x\left[\int_0^{ t\wedge\tau}h(X_u)\eta\left(t-u,X_u\right)du\right].
\]
By \eqref{eqn:boundnessCondition} we immediately have
\begin{equation}
\left|\tilde{q}_\tau(t)\right|\le M\epsilon ^2.
\end{equation}
For $\beta\in \mathbb{C}$ and $\text{Real}(\beta)\ge 0$, consider the Laplace transform of $\tilde{q}_\tau(t)$
\[
\LT\left\{\tilde{q}_\tau(t)\right\}(\beta)=\epsilon^{2}\int_0^\infty e^{-\beta t} \mathbb{E}_x\left[\int_0^{ t\wedge\tau}h(X_u)\eta\left(t-u,X_u\right)du\right]dt.
\]
Based on \eqref{eqn:boundnessCondition} and the dominated convergence theorem, we change the order of integral and expectation
\begin{equation}\label{eqn:LT1}
\LT\left\{\tilde{q}_\tau(t)\right\}(\beta)=\epsilon^{2}\mathbb{E}_x\left[\int_0^\infty \int_0^{ t\wedge\tau}e^{-\beta t} h(X_u)\eta\left(t-u,X_u\right)dudt\right].
\end{equation}
In addition, \eqref{eqn:boundnessCondition} also gives the Fubini's theorem so

\begin{align*}
\int_0^\infty \int_0^{ t\wedge\tau}e^{-\beta t}h(X_u)\eta\left(t-u,X_u\right)du dt&=\int_0^\infty \int_0^{\tau}\bold{1}_{\left\{u\le  t\right\}}e^{-\beta t}h(X_u)\eta\left(t-u,X_u\right)du dt\\
&= \int_0^{\tau}\int_0^\infty\bold{1}_{\left\{u\le  t\right\}}e^{-\beta t}h(X_u)\eta\left(t-u,X_u\right) dtdu\\
&=\int_0^\tau \LT\set{\bold{1}_{\left\{u\le  t\right\}}\eta\left(t-u,X_u\right)}(\beta)h(X_u) du.
\end{align*}
Note $\bold{1}_{\set{u\le t}}$ is the indicator function, which can also be written as the Heaviside step function $H_u(t)$. Consider the fact \cite{oberhettinger2012tables} that
\[
\LT\set{H_u(t)\eta(t-u,x)}(\beta)=e^{-\beta u}\LT\set{\eta(t,x)}(\beta),
\]
where by our notation
\[
\LT\set{\eta(t,x)}(\beta)=\LT\set{\ILT\set{\frac{\partial}{\partial x}f_1(\beta,x)}(t)}(\beta)=\frac{\partial}{\partial x}f_1(\beta,x).
\]
Therefore \eqref{eqn:LT1} can be re-expressed as
\begin{equation}
\label{LTform}
\LT\left\{\tilde{q}_\tau(t)\right\}(\beta)=\epsilon^{2}\mathbb{E}_x\left[\int_0^\tau e^{-\beta u}h(X_u)\frac{\partial }{\partial x}f_1(\beta,X_u)du\right].
\end{equation}

In the next step we show $\LT\left\{\tilde{q}_\tau(t)\right\}(\beta)$ indeed is the LT for the error function $q_{\tau}(t)$. Then the uniqueness of inverse LT concludes our proof. To see this, let 
\begin{equation}\label{eqn:Q1}
Q( \beta,x):=f(\beta,x)-f^1(\beta,x),
\end{equation}
where follow similar convention $f^1(\beta,x):=f^1(x)$ and $f^1(x)$ is as introduced in \eqref{eqn:totaltransform}. Note that $f(\beta,x)=\LT\set{\mathbb{P}_x\set{\tau\in dt}}(\beta)$ and $f^1(\beta,x)=\LT\set{\mathbb{P}^1_x\set{\tau^*\in dt}}(\beta)$. By the linearity of LT we therefore have
\begin{equation}\label{eqn:Q2}
Q( \beta,x)=\LT\left\{q_\tau(t)\right\}(\beta).
\end{equation}
As $f,f^1\in C^2$, so is $Q( \beta,x)\in C^2$. Apply the infinitesimal generator $\OperatorA$ on $Q(\beta,x)$. Then by \eqref{eqn:FPTode}, \eqref{eqn:o1} and \eqref{eqn:o2} with $i=1$, after standard calculations we get
\begin{equation}\label{eqn:ODEerror}
\mathcal{A}Q-\beta Q =-\epsilon^{2}h f_1^{'},\ x\in \mathcal{D}^u.
\end{equation}
Note \eqref{eqn:ODEerror} is an equation about $x$ and $f_1^{'}$ is a short for $\frac{\partial }{\partial x}f_1(\beta,x)$. Since $f$ and $f^1$ share the same boundary conditions, so the boundary condition of ODE \eqref{eqn:ODEerror} is given by
\begin{equation}\label{eqn:ODEerrorboundary}
Q\left(\partial \mathcal{D}^u\right)=(0,0)^T.
\end{equation}
According to \cite{peskir2006optimal}, the boundary value problem \eqref{eqn:ODEerror} and \eqref{eqn:ODEerrorboundary} has the following unique solution 
\[
Q(\beta,x)=\epsilon^{2}\mathbb{E}_x\left[\int_0^\tau e^{-\beta u}h(X_u)\frac{\partial }{\partial x}f_1(\beta,X_u)du\right].
\]
The uniqueness in ODE solution and the uniqueness in inverse LT indicates
\[
q_\tau(t)=\tilde{q}_\tau(t).
\]

\end{proof}

\begin{rmk}
Proposition \ref{prop:411} shows the error bound is uniformly valid on $t\ge 0$. When $\epsilon\downarrow 0^+$, the error converges to $0$. This is true as when $X_t\rightarrow W_t$ pathwisely, referring to Proposition \ref{cor42} we have $\mathbb{P}^1_x\left(\tau^*\in dt\right)\rightarrow p_0(t),\ \forall t\ge 0$.
\end{rmk}

\begin{rmk}
On the other hand, the conclusion in Proposition \ref{prop:411} does not restrict applying perturbation for $\epsilon >1$. In fact, an exact error function is given and we can estimate the error level via simulation. Even in the case that $\epsilon>1$, there are possibilities that 
\[
\left|\mathbb{E}_x\left[\int_0^{t\wedge \tau}\left(e^{-2\alpha X_u}-c\right)\eta(t-u,X_u)du\right]\right|<<\frac{1}{\epsilon^2},\ t\in (0,+\infty).
\]
\end{rmk}

\section{Model Implementation}
\subsection{Extended SDE with Constant Volatility}\label{sec51}
For practical purpose it is more interesting to take the volatility into account. We extend SDE \eqref{eqnexp} by adding a constant volatility $\sigma>0$:
\begin{equation}\label{eqn:eqxexpestended}
dX_t=\epsilon(e^{-2\alpha X_t}-c)dt+\sigma dW_t,\ X_0=x\in \mathbb{R}.
\end{equation}
Introduce the scaled version of $\set{X_t}_{t\ge 0}$ and define $\set{\tilde{X}_t}_{t\ge 0}$ as
\[
\tilde{X}_t:=\frac{X_t}{\sigma}.
\]
Then by setting $\tilde{\epsilon}:=\frac{\epsilon}{\sigma}$, $\tilde{\alpha}:=\alpha\sigma$ and $\tilde{x}=\frac{x}{\sigma}$, we see $\set{\tilde{X}}_{t\ge 0}$ indeed is the diffusion process described by SDE \eqref{eqnexp}: 
\[
d\tilde{X}_t=\tilde{\epsilon}\brck{e^{-2\tilde{\alpha}\tilde{X_t}}-c}dt+dW_t,\ \tilde{X}_0=\tilde{x}\in \mathbb{R}.
\]
In addition, for $a\le x$, by letting $\tilde{a}=\frac{a}{\sigma}$ the result of running minimum in Corollary \ref{cor:runningMinimum} can be extended accordingly.

\subsection{Model Calibration}
In this section we provide a calibration scheme for the extended SDE \eqref{eqn:eqxexpestended}. Denote the observations of asset prices $\set{P_t}_{t=0,1,...,N}$ by
\begin{equation}\label{eqn:logtrans}
P_t=P_0e^{\hat{X_t}},\ t=0,1,...,N.
\end{equation}
Then $\set{\hat{X_t}}_{t=0,...,N}$ represents the normalised log-price with $\hat{X}_0=0$. Let $\set{\hat{r}_t}_{t=1,...,N}$ to be the log-return of $\set{P_t}_{t=0,1,...,N}$. By definition we have
\begin{equation}\label{en:return}
\hat{r}_t=\hat{X}_t-\hat{X}_{t-1},\ t=1,...,N.
\end{equation}

Consider the calibration based on $\set{\hat{r}_t}_{t=1,...,N}$. Mathematically, there are 4 parameters to be decided. Therefore at least 4 different statistical quantities should be provided. A natural candidate is the first four moments of $\set{\hat{r}_t}_{t=1,...,N}$. However, on the one hand, as we discussed in Section \ref{sec2}, the bubble dynamic in different regimes could have totally different statistical behaviours. So global moments on the whole time-series may not be representative. On the other hand, from Proposition \ref{prop:density}, $\set{X_t}_{t\ge 0}$ has a very complicated probability density. Following the proposition, we cannot easily get the explicit expression even for the first moment. Instead of using traditional moments calibration, we provide an alternative scheme with the piecewise time-series under different bubble regimes.

Recall those three regimes of $\set{X_t}_{t\ge 0}$ in a bubble cycle, according to which we make the following assumptions:
\begin{itemize}
\item Regime I), displacement. During this period we assume $X_t\approx 0$. SDE \eqref{eqn:eqxexpestended} then can be simplified as
\begin{equation}\label{eqn: Regime I SDE}
dX_t\approx \epsilon(1-c) dt+\sigma dW_t.
\end{equation}
\item Regime II), boom. In this stage the dynamic follows SDE \eqref{eqn:eqxexpestended} but will visit the equilibrium level. Denote the level by $X^{R}$, we have
\begin{equation}\label{eqn:Regime II Drift}
e^{-2\alpha X^{R}}=c.
\end{equation}
\item Regime III) euphoria (\& profit taking). Within these two steps $X_t$ hits the record-high level. Assume $e^{-2\alpha X_t}\approx 0$ then SDE \eqref{eqn:eqxexpestended} degenerates to
\begin{equation}\label{eqn:Regime III SDE}
dX_t\approx -c\epsilon dt+\sigma dW_t.
\end{equation}
\end{itemize}  

Besides, we further assume each regime could be recognised from the data. Let $\set{0,1,...,t_1}$, $\set{t_1,...,t_2}$, $\set{t_2,...,t_3}$ to be the time periods for regimes I, II and III. Then denote the piecewise time-series in each regime by
\[
\hat{X}^I:=\left\{\hat{X}_t\right\}_{t=0,...,t_1},\ \hat{X}^{II}:=\left\{\hat{X}_t\right\}_{t=t_1,...,t_2},\ \hat{X}^{III}:=\left\{\hat{X}_t\right\}_{t=t_2,...,t_3}.
\] 
The corresponding time-series for log-returns are given by 
\[
\hat{r}^I:=\left\{\hat{r}_t\right\}_{t=1,...,t_1},\ \hat{r}^{II}:=\left\{\hat{r}_t\right\}_{t=t_1+1,...,t_2},\ \hat{r}^{III}:=\left\{\hat{r}_t\right\}_{t=t_2+1,...,t_3}.
\] 
Also assume that the equilibrium level is observable and denote the observation by
\[
\hat{X}^{R}.
\]

We now consider parameter estimates. Start with $\hat{\epsilon}$ and $\hat{c}$. The general idea is to take the expected log-returns from regimes I and III in to account. Let 
\[
\bar{r}^I:=\text{Mean}\left(\hat{r}^I\right) \text{ and } \bar{r}^{III}:=\text{Mean}\left(\hat{r}^{III}\right)
\]
to be the annualized sample averages of returns. By matching the sample means with theoretical expectations from $dX_t$ in \eqref{eqn: Regime I SDE} and \eqref{eqn:Regime III SDE}, we have the following equations 
$$
\begin{cases}
\hat{\epsilon}(1-\hat{c})=\bar{r}^I\\
-\hat{\epsilon}\hat{c}=\bar{r}^{III}
\end{cases}.
$$
Solving the equations we get
\begin{equation}
\label{epsilonc}
\begin{cases}
\hat{\epsilon}=\bar{r}^I-\bar{r}^{III}\\
\hat{c}=-\frac{\bar{r}^{III}}{\bar{r}^I-\bar{r}^{III}}
\end{cases}.
\end{equation}

\begin{rmk}
Note that according to our assumptions, regime I should provide positive trend ($\bar{r}^I\ge 0$) while regime III generates negative moves ($\bar{r}^{III}\le 0$). Therefore $\hat{\epsilon}$ and $\hat{c}$ are guaranteed to be positive. Moreover, since 
\[
0\le -\bar{r}^{III}\le \bar{r}^I-\bar{r}^{III},
\]
so $0\le \hat{c}\le 1$.
\end{rmk}

\begin{rmk}
In order to have a more effective calibration, in $\bar{r}^I$ and $\bar{r}^{III}$ estimations we can (*) take the average of only positive returns in regime I and only the negative returns in regime III. In addition, we are more interested in the longer term trend rather than the daily trend. So (**) using monthly rolling returns would help in enhancing the estimation stability. We add (*) and (**) as special data cleaning treatments in our algorithm.
\end{rmk}

Consider $\hat{\sigma}$. By observing \eqref{eqn: Regime I SDE} and \eqref{eqn:Regime III SDE} we see the volatilities in $\hat{r}^I$ and $\hat{r}^{III}$ are provided by the Brownian motion part only. Let $\hat{r}^{I\&III}:=\hat{r}^I\cup \hat{r}^{III}$. Then we can compute $\hat{\sigma}$ by
$$
\hat{\sigma}=StdDev\left({\hat{r}}^{I\&III}\right).
$$
As an alternative plan, notice that usually regime III has more volatile time-series. Therefore in order to capture a more significant volatility we choose to use $\hat{r}^{III}$ only:
\begin{equation}
\label{sigma}
\hat{\sigma}=StdDev\left({\hat{r}}^{III}\right).
\end{equation} 

Given $\hat{c}$, the last parameter $\hat{\alpha}$ is easy to compute. Based on \eqref{eqn:Regime II Drift}, we immediately have
\begin{equation}
\label{alpha}
\hat{\alpha}=-\frac{\ln\left(\hat{c}\right)}{2\hat{X}^R}.
\end{equation}

We summarise the calibration algorithm in Algorithm \ref{alg:1}. 

\begin{algorithm}
  \caption{$\set{X_t}_{t\ge 0}$ Parameter Calibration }
\label{alg:1}
\begin{enumerate}
\item Determine the time ranges for regimes I-III, and correspondingly calculate the floored log-price $\hat{X}^{I},\ \hat{X}^{II},\ \hat{X}^{III}$ by \eqref{eqn:logtrans}. Identify the equilibrium level $\hat{X}^{R}$.
\item Calculate the monthly rolling log-returns of ${\hat{r}}^I_{m}$ and ${\hat{r}}^{III}_{m}$ from $\hat{X}^{I}$ and $\hat{X}^{III}$ respectively. Calculate ${\bar{r}}^I$ and ${\bar{r}}^{III}$ via
\[
\begin{cases}
\bar{r}^I=Mean\left({\hat{r}}^I_{m}\bigg|{\hat{r}}^I_{m}\ge0\right)\times 12\\
\bar{r}^{III}=Mean\left({\hat{r}}^{III}_{m}\bigg|{\hat{r}}^{III}_{m}\le 0\right)\times 12
\end{cases},
\]
and use Equation \eqref{epsilonc} to calibrate $\hat{\epsilon},\ \hat{c}$.
\item Calculate the daily log-return time-series ${\hat{r}}^{III}_{d}$ from $\hat{X}^{III}$. Compute annualised return ${\hat{r}}^{III}$ via
\[
{\hat{r}}^{III}={\hat{r}}^{III}_{d}\times \sqrt{260},
\]
and calibrate $\hat{\sigma}$ using Equation \eqref{sigma}.
\item  Substitute $\hat{X}^R$ from step 1 and $\hat{c}$ from step 2 into equation \eqref{alpha} to calibrate $\hat{\alpha}$.
\end{enumerate}
\end{algorithm}

\begin{rmk}
We need to highlight that Algorithm \ref{alg:1} relies on two judgmental decisions, i.e. 1) the time range for different regimes and 2) the equilibrium level. During an asset price increasing period (displacement, boom, euphoria), from the economical point of view it is not difficult to differentiate those three regimes. Even when there is no clear economical signal, we can still split the time-series equally into three pieces. However, for $\hat{X}^R$, without a significant price drop, mathematically it is very challenging to decide the equilibrium level. Therefore fundamental analysis from economics may be required. The enhancement of Algorithm \ref{alg:1} will be remained in the future work.
\end{rmk}

\subsection{Numerical Examples}
We provide three numerical examples. The first two exercises are similar in nature, where based on historical data we verified the effectiveness of $\set{X_t}_{t\ge 0}$ in capturing bubble dynamics. In the third exercise we predicted drop-down probabilities for BitCoin.

\subsubsection{1997-01-02 to 2003-12-30 NASDAQ Composite Index} \label{exmp1}
The US dot-com bubble \cite{cassidy2003dotcom} could be observed from the technology-dominated NASDAQ Composite Index (US ticker symbol \^{ }IXIC). From mid 90's, \^{ }IXIC grew exponentially from below 1,000 USD to about 5,000 USD. The index hit its historical maximum in 2000-03-10, and at which date the total trading amount exceeded 10 Trillion USD (according to Yahoo Finance). After then the market collapsed rapidly and dropped back to about 1,000 USD in 2002.

In this exercise we used the adjusted daily close price of \^{ }IXIC from 1997-01-02 to 2003-12-30. The data was downloaded from Yahoo Finance. Note that, for the purpose of burst time prediction, there is no sense to calibrate the model using the full-cycle data. Therefore only a truncated time-series was used in model calibration. To be more specific, we chose the calibration regimes as follows
\begin{equation}\label{ITAX_raw_ts}
\begin{cases}
\hat{X}^{I}:\ \text{1997-01-02 ({$P_0=1,280$}) to 1997-06-26 ({$P_{t_1}=1,436$});} \\
\hat{X}^{II}:\ \text{1997-06-26 ({$P_{t_1}=1,436$}) to 1999-02-10 ({$P_{t_2}=2,309$})}; \\
\hat{X}^{III}:\ \text{1999-02-10 ({$P_{t_2}=2,309$}) to 2000-10-18 ({$P_{t_3}=3,171$}).}
\end{cases}
\end{equation}
The red curve in Figure \ref{icixCalibration} plots the full series of $\set{\hat{X}_t}_{t=0,1,...,N}$. By observation we set the equilibrium level to be $\hat{X}^R=0.67\ ({P_{R}=2,502})$. 

In order to compare our new model with existing models, we also included the OU process and drifted Brownian motion (DBM). The time-series used for calibrating these two models were the same as in $\set{X_t}_{t\ge 0}$ calibration, i.e. from 1997-01-02 to 2000-10-18. For the MLE OU calibration algorithm, cf. \cite{smith2010simulation}. We estimated the mean and volatility directly in the DBM. 1,000 paths between 1997-01-02 and 2003-12-30 were simulated by three different models. In Figure \ref{icixCalibration}, apart from the historical price of \^{ }IXIC,  we demonstrate the best path among 1,000 simulations for each model. It is clear by the graph that our new model provides better fit than existing models. To quantitatively see the closeness of different paths to the historical dynamic, the correlations for each model were calculated:
\[\set{X_t}_{t\ge 0}:\ 91.20\%,\ OU:\ 81.01\%,\ DBM:\ 72.03\%.\]
As expected, $\set{X_t}_{t\ge 0}$ provided the highest correlation while DBM was the worst among three models. To further explain our algorithm, we plot calibration regimes in Figure \ref{icixregion}. We also show 10,000 onward simulation paths for $\set{X_t}_{t\ge 0}$ with $X_0=\hat{X}_{t_3}$. From the figure we see the historical prices are fully covered by the simulation paths. This indicates our model is effective.

\begin{figure}[h]
    \centering
    \includegraphics[height = 0.475\textwidth, width=.7\textwidth]{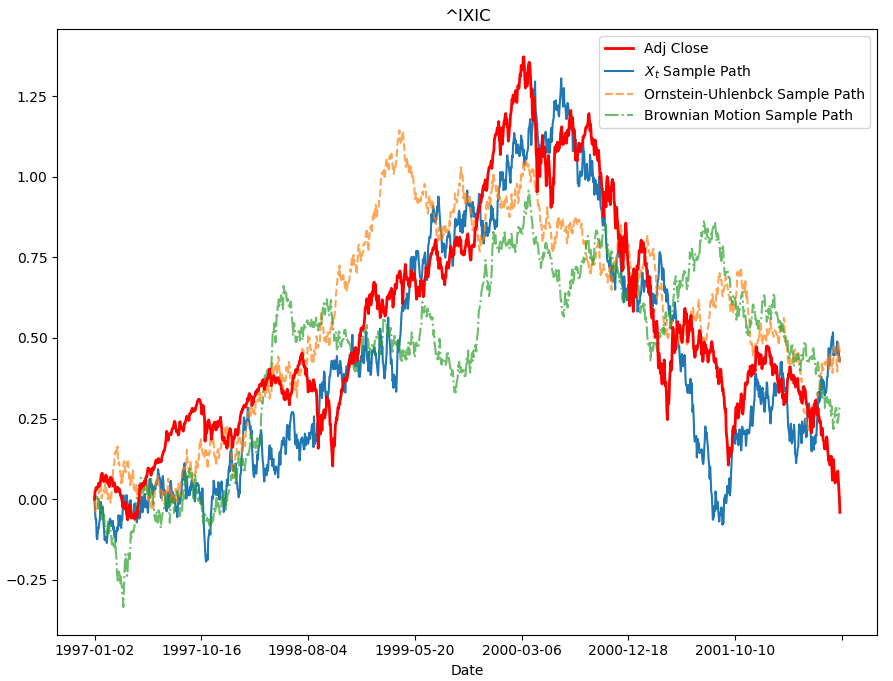}
    \caption{Model calibration comparisons for NASDAQ index (US ticker symbol \^{ }IXIC)  from 1997-01-02 to 2003-12-30. Red curve: historical adjusted log-price; blue curve: the best of 1,000 simulations from $\set{X_t}_{t\ge 0}$; orange curve: the best of 1,000 simulations from OU process; green curve: the best of 1,000 simulations from DBM. Calibration parameters, $\set{X_t}_{t\ge 0}:\ (\hat{\epsilon},\hat{\alpha},\hat{\sigma},\hat{c})=(0.39, 0.23, 0.43, 0.73);\ OU:\ (\hat{\kappa},\hat{\mu},\hat{\sigma})=(0.47, 1.09, 0.31);\ BM:\ (\hat{\mu},\hat{\sigma})=(0.25, 0.31)$. The data source is from Yahoo Finance.
       }
    \label{icixCalibration}
\end{figure}

\begin{figure}[h]
    \centering
    \includegraphics[height = 0.475\textwidth, width=.7\textwidth]{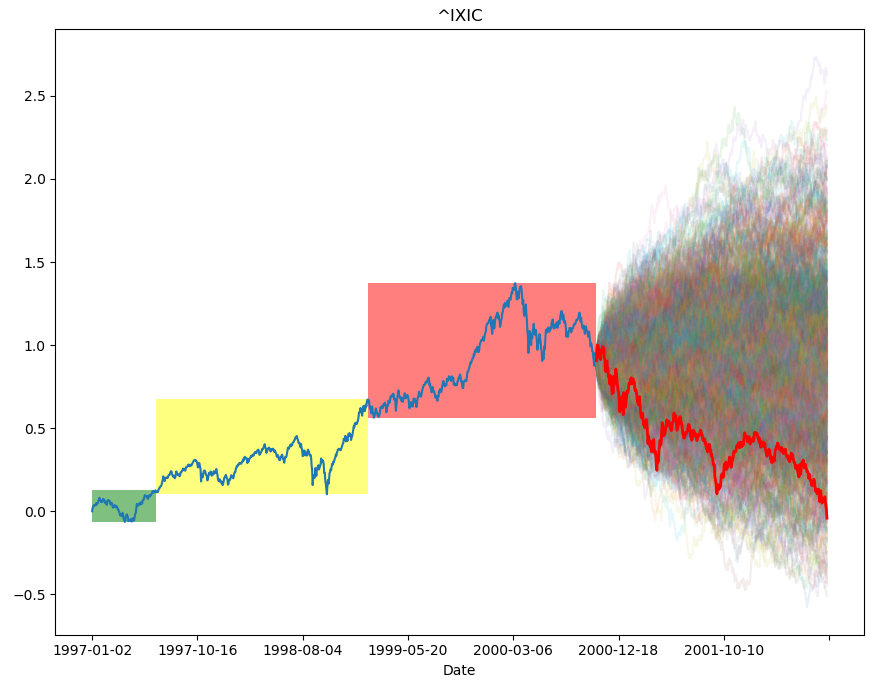}
    \caption{Algorithm \ref{alg:1} illustration based on \^{ }IXIC and 10,000 paths simulation starting from $X_0=\hat{X}_{t_3}$ . Green zone indicates regime I, the displacement stage; yellow zone indicates regime II, the boom stage; red zone indicates regime III, the euphoria \& profit taking stages. Blue curve shows the historical data used for calibration. Red curve, covered by shadowed region, shows the historical data after $t_3$. The shadowed region plots 10,000 simulation paths.}
    \label{icixregion}
\end{figure}

\clearpage

\subsubsection{2006-01-04 to 2008-12-31 Shanghai Stock Exchange Composite Index}\label{exmp2}
We use a second example to confirm our observations from Section \ref{exmp1}. The 2007 Chinese stock market crash \cite{jiang2010bubble} just happened before the 2008 global financial crisis. Starting in early 2006, the Shanghai Stock Exchange Composite Index (US ticker symbol SSEC, China ticker symbol 000001.SS) increased from about 1,000 CNY to 6,092 CNY in mid-October, 2007. And within one year's time, from October 2007 to October 2008, the price dropped below 1,800 CNY. Similar to the pattern in \^{ }IXIC, the historical log-price of SSEC dropped rapidly after the sudden peak, and before which there was a sharp increase.

The exercise settings were the same as in Section \ref{exmp1}. We only mention the regime settings and make comments where are necessary. 
\begin{equation}\label{SSEC_raw_ts}
\begin{cases}
\hat{X}^{I}:\ \text{2006-01-04 ({$P_0=1,180$}) to 2006-03-06 ({$P_{t_1}=1,288$});} \\
\hat{X}^{II}:\ \text{2006-03-06 ({$P_{t_1}=1,288$}) to 2007-05-30 ({$P_{t_2}=4,053$}); equilibrium level }\hat{X}^R=1.23\ ({P_{R}=4,040})\text{;}\\
\hat{X}^{III}:\ \text{2007-05-30 ({$P_{t_2}=4,053$}) to 2008-04-21 ({$P_{t_3}=3,116$}).}
\end{cases}
\end{equation}

Figure \ref{sse100Calibration} demonstrates comparisons between best simulation paths and historical log-price. We can immediately see that the OU process provided a much faster mean-reversion rate than which was reflected by the price dynamic. This shows the OU process cannot provide enough degree of freedom in calibrating bubble dynamics. The correlations for different models to the actual data were given by:
$$
\set{X_t}_{t\ge 0}:96.11\%,\ OU:88.00\%,\ DBM:84.42\%.
$$
Similar plot for the algorithm illustration and 10,000 simulation paths is given in Figure \ref{sse100region}. Through this exercise we further confirm that our new model is a good candidate for describing economic bubbles.

\begin{figure}[h]
    \centering
    \includegraphics[height = 0.461\textwidth, width=.7\textwidth]{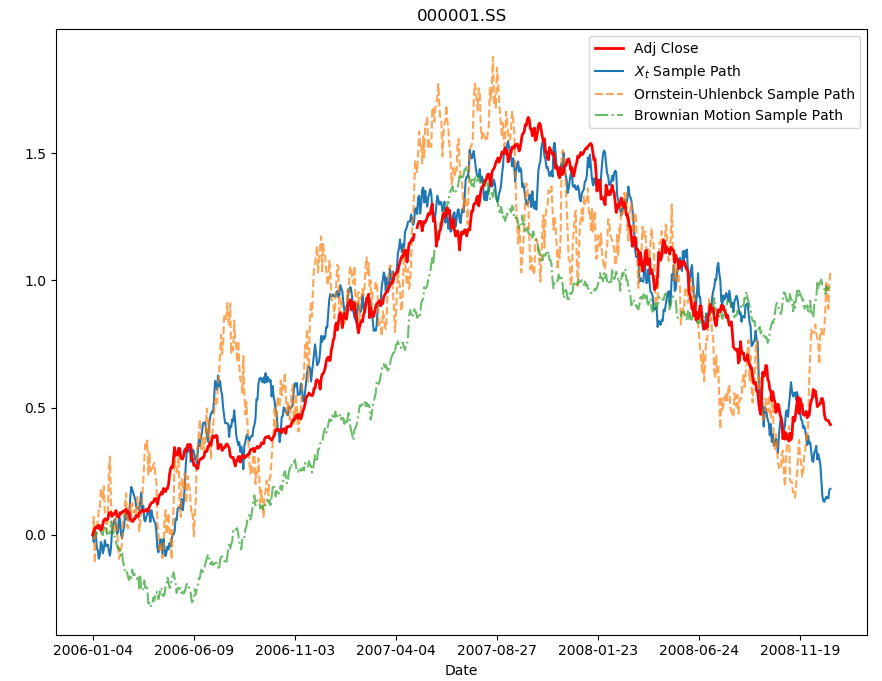}
    \caption{Model calibration comparisons for Shanghai Stock Exchange Composite index (US ticker symbol SSEC, China ticker symbol 000001.SS) from 2006-01-04 to 2008-12-31. Red curve: historical adjusted log-price; blue curve: the best of 1,000 simulations from $\set{X_t}_{t\ge 0}$; orange curve: the best of 1,000 simulations from OU process; green curve: the best of 1,000 simulations from DBM. Calibration parameters, $\set{X_t}_{t\ge 0}:\ (\hat{\epsilon},\hat{\alpha},\hat{\sigma},\hat{c})=(0.32, 0.14, 0.56, 0.70);\ OU:\ (\hat{\kappa},\hat{\mu},\hat{\sigma})=(3.30, 0.97, 1.20);\ BM:\ (\hat{\mu},\hat{\sigma})=(0.44, 0.33)$. The data source is from Yahoo Finance.
    }
    \label{sse100Calibration}
\end{figure}

\begin{figure}[h]
    \centering
    \includegraphics[height = 0.461\textwidth, width=.7\textwidth]{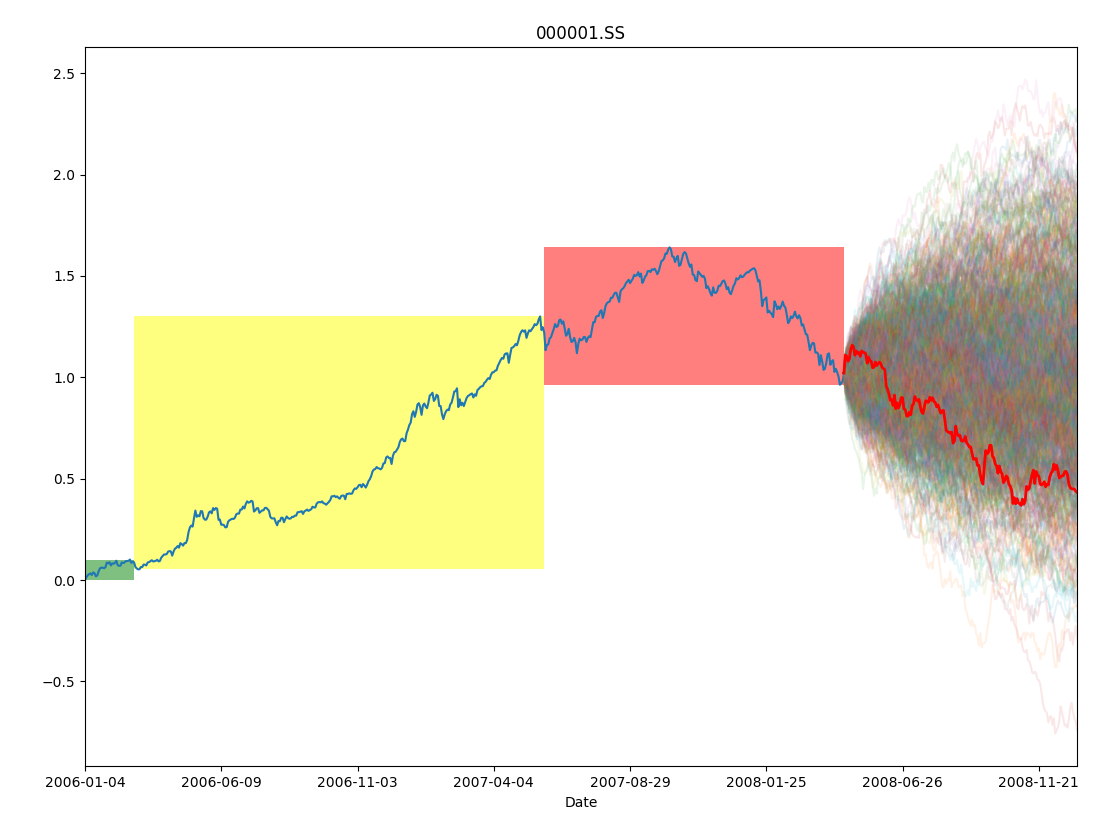}
    \caption{Algorithm \ref{alg:1} illustration based on 000001.SS and 10,000 paths simulation starting from $X_0=\hat{X}_{t_3}$. Green zone indicates regime I, the displacement stage; yellow zone indicates regime II, the boom stage; red zone indicates regime III, the euphoria \& profit taking stages. Blue curve shows the historical data used for calibration. Red curve, covered by shadowed region, shows the following historical data after $t_3$. The shadowed region plots 10,000 simulation paths.}
    \label{sse100region}
\end{figure}

\clearpage

\subsubsection{BitCoin Downward Probability Estimation}\label{exmp3}

2017 is a year of BitCoin. At the first trading day of 2017, the price of 1 BitCoin was 995.44 USD. Although spending 1,000 dollars to buy one cryptocurrency was unbelievable to people, within 1 year's time, the price hit 19,345.49 USD. Figure \ref{fig1bitcoin} illustrates patterns for the price and trading volume between 2016-01-01 and 2017-12-10. There are many potential reasons that driven the nearly 20 times increase, for example, the increasing investments from institutional investors, more open mind from lawmakers, etc.. We are interested in knowing whether the price would drop in the near future. In this exercise we conducted analysis in predicting the minimum of BitCoin price in the following month, effectively from 2017-12-10 to 2018-01-12.

\begin{figure}[h]
 \centering
        \includegraphics[width=.75\textwidth, height = .5\textwidth]{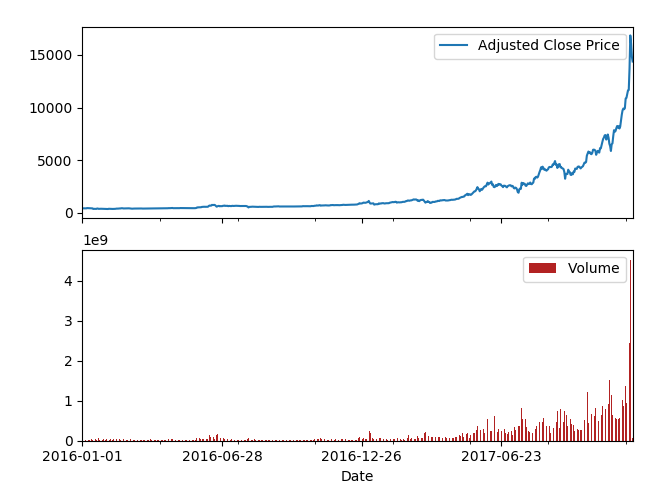}
        \caption{Bitcoin historical daily price and trading volume from 2016-01-01 to 2017-12-10. The data source is from Yahoo Finance.}
        \label{fig1bitcoin}
 \end{figure}
 
The model calibration was based on time-series between 2016-01-01 and 2017-12-10:
\begin{equation}\label{BitCoin_raw_ts}
\begin{cases}
\hat{X}^{I}:\ \text{2016-01-01 ({$P_0=433$}) to 2016-05-30 ({$P_{t_1}=528$});} \\
\hat{X}^{II}:\ \text{2016-05-30 ({$P_{t_1}=528$}) to 2017-08-13 ({$P_{t_2}=4,327$}); equilibrium level }\hat{X}^R=2.30\ ({P_{R}=4,327})\text{;}\\
\hat{X}^{III}:\ \text{2017-08-13 ({$P_{t_2}=4,327$}) to 2017-12-10 ({$P_{t_3}=14,371$}).}
\end{cases}
\end{equation}
Without mentioning too much detail, we summarise Algorithm \ref{alg:1} outputs in below
\begin{equation}\label{eqn:extendedpara}
\hat{\epsilon}=0.51;\ \hat{\alpha}=0.08; \hat{\sigma}=0.91; \hat{c}=0.69.
\end{equation}

The prediction was made on 2017-12-10 with the price at $P_{t_3}=14,371$\footnote{Note that, the data in our record does not correspond to the close price of 2017-12-10. In fact, the data was downloaded when the market was still under trading. }. We considered $0\%$ to $60\%$ drops from $P_{t_3}$. To convert the drop percentages from $\set{P_t}_{t=0,1,...,N}$ to the log-price space, $\set{X_t}_{t=0,1,...,N}$, we calculated the hitting level $a$ via \eqref{eqn:logtrans}. Referring to Section \ref{sec51}, we transferred parameters in \eqref{eqn:extendedpara}, together with $a$, from the extended SDE \eqref{eqn:eqxexpestended} to the parameters in the standard SDE \eqref{eqnexp}. By Corollary \ref{cor:runningMinimum}, in the end we were able to have the probability distribution for the minimum price within one month time. On the other hand, an error evaluation on the perturbed FPTD should be given. Refer to Proposition \ref{prop:411}. The relative error is given by
\[
e(t):=\left|\frac{P\brck{\tau^*\in dt}}{P\brck{\tau^*\in dt}+q_\tau(t)}\right|.
\]
Using the probabilistic representation, we estimated $q_\tau(t)$ via 10,000 paths simulation. It should be noticed that, the relative error generally is high at tails as the actual density converges to $0$. Therefore it is not necessary to compute the relative error at each point. In fact, we are more concerned that whether the peak of the distribution would be changed by perturbations. So only relative errors on the density peak were computed. Table \ref{table1} summarises the results.

\begin{table}[h]
\centering
\begin{tabular}{|c|c|c|c|c}
\hline
  Percentage of Drop  &   Price $P_l$ (USD)&  Probability $\mathbb{P}\brck{P_t^*\le P_l}$ &  Peak Relative Error \\
  \hline
  \hline
    0\% &  14,371.62 &     100.00\% &             0.00\% \\\hline
    5\% &  13,653.05 &     84.85\% &            4.97\% \\\hline
    10\% &  12,934.47 &    69.38\% &           1.48\% \\\hline
    15\% &  12,215.89 &     54.25\% &          0.90\% \\\hline
    20\% &  11,497.30 &     40.19\% &          0.04\% \\\hline
    25\% &  10,778.72 &     27.88\% &          0.59\% \\\hline
    30\% &  10,060.14 &     17.87\% &          0.86\% \\\hline
    35\% &   9,341.56 &     10.38\% &           0.88\% \\\hline
    40\% &   8,622.98 &     5.35\% &             0.68\% \\\hline
    45\% &   7,904.40 &     2.37\% &             1.69\% \\\hline
    50\% &   7,185.81 &     0.86\% &             1.45\% \\\hline
    55\% &   6,467.23 &     0.25\% &             2.40\%\\\hline
    60\% &   5,748.65 &     0.05\% &             1.78\% \\\hline
\end{tabular}
\caption{BitCoin downward price prediction between 2017-12-10 and 2018-01-12. Columns 1-4 correspond to the percentages of price drop, dropped price $P_l$, probability of the lowest price $P_t^*$ below $P_l$, and the relative error in density peaks.}
\label{table1}
\end{table}

First by checking the last column (relative errors), we see in general the perturbation model is accurate. The largest error was in the 5\% drop. In this case the hitting level is very close to the initial price $P_{t_3}$. As a result, the density curve will shrink to the y-axis. Therefore a larger error is expected. Analogously, large errors might also exist in the case that hitting levels are far to the initial price. By ruling out the extreme drops, in the range of 10\% to 50\%, we see the estimation errors remained below 2\%. 

We now consider the possibility of market collapse. Referring to the scenarios in \^{ }IXIC and 000001.SS, we found their largest drops in a month were about 30\%, and which happened in the spring of 2000 and autumn of 2008, respectively. Then check the probability of 30\% drop for BitCoin. From Table \ref{table1} we only see about 17.87\%. In fact, even for a 20\% drop, the probability was about 40.19\%. This means there was more than half chance that the price would remain above 11,497.30. Therefore we concluded that the market was unlikely to collapse in the next month. 

We collected the one month data from 2017-12-10 to 2018-01-12 and plot the time-series in Figure \ref{fig4bitcoin}. From the graph we see the lowest close price was 12,531.52 on 2017-12-30. This verified our conclusion that the market would not collapse immediately. On the other hand, compare the probabilities in Table \ref{table1} with the thresholds in Figure \ref{fig4bitcoin}. The price on 2017-12-30 broke the 10\% drop threshold, where the probability given by our prediction was 69.38\%. This further confirmed that the model is effective.

\begin{figure}[h]
\centering
        \includegraphics[width=.7\textwidth, height = .5\textwidth]{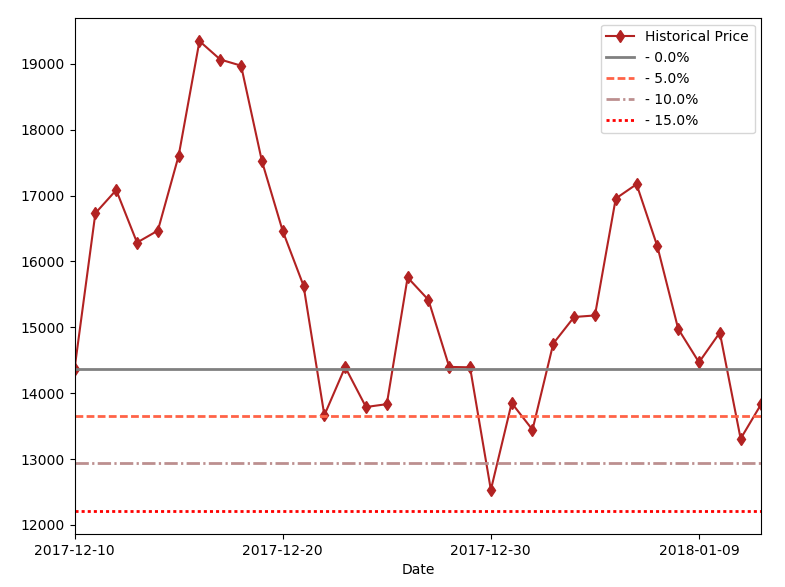}
        \caption{BitCoin close price between 2017-12-10 and 2018-01-12. The probabilities of different thresholds are reported in Table \ref{table1}.}
        \label{fig4bitcoin}
 \end{figure}

\section{Conclusion}
In this paper we find a new diffusion process which can be used in modelling economic bubbles. The simple form of the model enables us to deduce its downward FPTD explicitly. Therefore the paper provides a useful tool in estimating the burst time of an economic bubble. Numerical examples in Section 5 consistently confirm that the model and its prediction are effective. Results in Section 3 show the process has desirable properties which potentially can be employed in the future option pricing work. The perturbation technique, as introduced in Section 4, can be extended in finding explicit FPTDs of other diffusion processes. In another working paper of us, the corresponding closed-form FPTDs have been found for the OU and Bessel processes. One remaining issue is the exact simulation for the process. This requires further understandings to the $\theta\brck{r,s}$ function and we leave it for the future work.

\clearpage

\bibliographystyle{plain}
\bibliography{reff.bib}

\begin{thebibliography}{10}

\bibitem{abate2006unified}
Joseph Abate and Ward Whitt.
\newblock A unified framework for numerically inverting {L}aplace transforms.
\newblock {\em INFORMS Journal on Computing}, 18(4):408--421, 2006.

\bibitem{abramowitz1964handbook}
Milton Abramowitz and Irene~A Stegun.
\newblock {\em Handbook of mathematical functions: with formulas, graphs, and
  mathematical tables}, volume~55.
\newblock Courier Corporation, 1964.

\bibitem{barrieu2004study}
Pauline Barrieu, A~Rouault, and M~Yor.
\newblock A study of the {H}artman—{W}atson distribution motivated by
  numerical problems related to the pricing of {A}sian options.
\newblock {\em Journal of Applied Probability}, 41(4):1049--1058, 2004.

\bibitem{bateman1954tables}
Harry Bateman.
\newblock {\em Tables of integral transforms [volumes I \& II]}, volume~1.
\newblock McGraw-Hill Book Company, 1954.

\bibitem{borodin2012handbook}
Andrei~N Borodin and Paavo Salminen.
\newblock {\em Handbook of {B}rownian motion-facts and formulae}.
\newblock Birkh{\"a}user, 2012.

\bibitem{brooks2005three}
Chris Brooks and Apostolos Katsaris.
\newblock A three-regime model of speculative behaviour: modelling the
  evolution of the {S}\&{P} 500 composite index.
\newblock {\em The Economic Journal}, 115(505):767--797, 2005.

\bibitem{brooks2005trading}
Chris Brooks and Apostolos Katsaris.
\newblock Trading rules from forecasting the collapse of speculative bubbles
  for the {S}\&{P} 500 composite index.
\newblock {\em The Journal of Business}, 78(5):2003--2036, 2005.

\bibitem{cassidy2003dotcom}
John Cassidy.
\newblock dotcom: how {A}merica lost its mind and money in the internet era,
  2003.

\bibitem{chiani2003new}
Marco Chiani, Davide Dardari, and Marvin~K Simon.
\newblock New exponential bounds and approximations for the computation of
  error probability in fading channels.
\newblock {\em IEEE Transactions on Wireless Communications}, 2(4):840--845,
  2003.

\bibitem{cox2005local}
Alexander~MG Cox and David~G Hobson.
\newblock Local martingales, bubbles and option prices.
\newblock {\em Finance and Stochastics}, 9(4):477--492, 2005.

\bibitem{daalhuis2010confluent}
AB~Olde Daalhuis.
\newblock Confluent hypergeometric functions.
\newblock {\em NIST Handbook of Mathematical Functions, FWJ Olver, DW Lozier,
  RF Boisvert, and CW Clark, eds., Cambridge University, New York}, pages
  321--349, 2010.

\bibitem{dash2011tulipomania}
Mike Dash.
\newblock {\em Tulipomania: The story of the world's most coveted flower and
  the extraordinary passions it aroused}.
\newblock Hachette UK, 2011.

\bibitem{dassios2006square}
Angelos Dassios and Jayalaxshmi Nagaradjasarma.
\newblock The square-root process and {A}sian options.
\newblock {\em Quantitative Finance}, 6(4):337--347, 2006.

\bibitem{dassios2010perturbed}
Angelos Dassios and Shanle Wu.
\newblock Perturbed {B}rownian motion and its application to {P}arisian option
  pricing.
\newblock {\em Finance and Stochastics}, 14(3):473--494, 2010.

\bibitem{duck2009singular}
Peter~W Duck, Chao Yang, David~P Newton, and Martin Widdicks.
\newblock Singular perturbation techniques applied to multiasset option
  pricing.
\newblock {\em Mathematical Finance}, 19(3):457--486, 2009.

\bibitem{filimonov2017modified}
Vladimir Filimonov, Guilherme Demos, and Didier Sornette.
\newblock Modified profile likelihood inference and interval forecast of the
  burst of financial bubbles.
\newblock {\em Quantitative finance}, 17(8):1167--1186, 2017.

\bibitem{fouque2011multiscale}
Jean-Pierre Fouque, George Papanicolaou, Ronnie Sircar, and Knut S{\o}lna.
\newblock {\em Multiscale stochastic volatility for equity, interest rate, and
  credit derivatives}.
\newblock Cambridge University Press, 2011.

\bibitem{garber2001famous}
Peter~M Garber.
\newblock {\em Famous first bubbles: The fundamentals of early manias}.
\newblock mit Press, 2001.

\bibitem{hartman1974normal}
Philip Hartman and Geoffrey~S Watson.
\newblock " {N}ormal" distribution functions on spheres and the modified
  {B}essel functions.
\newblock {\em The Annals of Probability}, pages 593--607, 1974.

\bibitem{heston2006options}
Steven~L Heston, Mark Loewenstein, and Gregory~A Willard.
\newblock Options and bubbles.
\newblock {\em The Review of Financial Studies}, 20(2):359--390, 2006.

\bibitem{hladikova2012term}
Hana Hlad{\'\i}kov{\'a} and Jarmila Radov{\'a}.
\newblock Term structure modelling by using {N}elson-{S}iegel model.
\newblock {\em European Financial and Accounting Journal}, 7(2):36--55, 2012.

\bibitem{holt2009summary}
Jeff Holt.
\newblock A summary of the primary causes of the housing bubble and the
  resulting credit crisis: {A} non-technical paper.
\newblock {\em The Journal of Business Inquiry}, 8(1):120--129, 2009.

\bibitem{ishiyama2005methods}
Kazuyuki Ishiyama.
\newblock Methods for evaluating density functions of exponential functionals
  represented as integrals of geometric {B}rownian motion.
\newblock {\em Methodology and Computing in Applied Probability},
  7(3):271--283, 2005.

\bibitem{jiang2010bubble}
Zhi-Qiang Jiang, Wei-Xing Zhou, Didier Sornette, Ryan Woodard, Ken Bastiaensen,
  and Peter Cauwels.
\newblock Bubble diagnosis and prediction of the 2005--2007 and 2008--2009
  {C}hinese stock market bubbles.
\newblock {\em Journal of economic behavior \& organization}, 74(3):149--162,
  2010.

\bibitem{john2003dot}
Cassidy John.
\newblock Dot. con: How america lost its mind and money in the internet era,
  2003.

\bibitem{karatzas2012brownian}
Ioannis Karatzas and Steven Shreve.
\newblock {\em Brownian motion and stochastic calculus}, volume 113.
\newblock Springer Science \& Business Media, 2012.

\bibitem{kiselev2010simple}
Alexander Kiselev and Lenya Ryzhik.
\newblock A simple model for asset price bubble formation and collapse.
\newblock {\em arXiv preprint arXiv:1009.0299}, 2010.

\bibitem{lozier2003nist}
Daniel~W Lozier.
\newblock {NIST} digital library of mathematical functions.
\newblock {\em Annals of Mathematics and Artificial Intelligence},
  38(1-3):105--119, 2003.

\bibitem{matsumoto2005exponential}
Hiroyuki Matsumoto, Marc Yor, et~al.
\newblock Exponential functionals of {B}rownian motion, i: Probability laws at
  fixed time.
\newblock {\em Probability surveys}, 2:312--347, 2005.

\bibitem{minsky2008stabilizing}
Hyman~P Minsky and Henry Kaufman.
\newblock {\em Stabilizing an unstable economy}, volume~1.
\newblock McGraw-Hill New York, 2008.

\bibitem{oberhettinger2012tables}
Fritz Oberhettinger and Larry Badii.
\newblock {\em Tables of {L}aplace transforms}.
\newblock Springer Science \& Business Media, 2012.

\bibitem{olver2010nist}
Frank~WJ Olver.
\newblock {\em {NIST} handbook of mathematical functions hardback and CD-ROM}.
\newblock Cambridge university press, 2010.

\bibitem{peskir2006fundamental}
Goran Peskir.
\newblock On the fundamental solution of the {K}olmogorov—{S}hiryaev
  equation.
\newblock In {\em From stochastic calculus to mathematical finance}, pages
  535--546. Springer, 2006.

\bibitem{peskir2006optimal}
Goran Peskir and Albert Shiryaev.
\newblock {\em Optimal stopping and free-boundary problems}.
\newblock Springer, 2006.

\bibitem{phillips2015testing}
Peter~CB Phillips, Shuping Shi, and Jun Yu.
\newblock Testing for multiple bubbles: {L}imit theory of real-time detectors.
\newblock {\em International Economic Review}, 56(4):1079--1134, 2015.

\bibitem{reinhart20082007}
Carmen~M Reinhart and Kenneth~S Rogoff.
\newblock Is the 2007 {US} sub-prime financial crisis so different? {A}n
  international historical comparison.
\newblock {\em American Economic Review}, 98(2):339--44, 2008.

\bibitem{schrodinger1926quantisierung}
Erwin Schr{\"o}dinger.
\newblock Quantisierung als eigenwertproblem.
\newblock {\em Annalen der physik}, 385(13):437--490, 1926.

\bibitem{shiller2000irrational}
Robert~C Shiller.
\newblock Irrational exuberance.
\newblock {\em Philosophy \& Public Policy Quarterly}, 20(1):18--23, 2000.

\bibitem{shiryaev2002quickest}
Albert~N Shiryaev.
\newblock Quickest detection problems in the technical analysis of the
  financial data.
\newblock In {\em Mathematical Finance—Bachelier Congress 2000}, pages
  487--521. Springer, 2002.

\bibitem{shiryaev1961problem}
AN~Shiryaev.
\newblock The problem of the most rapid detection of a disturbance in a
  stationary process.
\newblock In {\em Soviet Math. Dokl}, volume~2, 1961.

\bibitem{smith2010simulation}
William Smith.
\newblock On the simulation and estimation of the mean-reverting
  {O}rnstein-{U}hlenbeck process.
\newblock {\em Commodities Markets and Modelling}, 2010.

\bibitem{soliverez1981general}
Carlos~E Soliverez.
\newblock General theory of effective hamiltonians.
\newblock {\em Physical Review A}, 24(1):4, 1981.

\bibitem{talay1994numerical}
Denis Talay.
\newblock Numerical solution of stochastic differential equations.
\newblock 1994.

\bibitem{wood1992bubble}
Christopher Wood.
\newblock {\em The bubble economy: the {J}apanese economic collapse}.
\newblock Sidgwick \& Jackson, 1992.

\bibitem{yor1992some}
Marc Yor.
\newblock On some exponential functionals of {B}rownian motion.
\newblock {\em Advances in applied probability}, 24(3):509--531, 1992.

\end{thebibliography}

\end{document}